\documentclass[letterpaper,11pt]{article} %
\usepackage[margin=1in]{geometry} %

\usepackage[utf8]{inputenc}
\usepackage[dvipsnames]{xcolor}
\usepackage{graphicx}
\usepackage{amsmath,amsfonts,amssymb,amsthm}
\usepackage{lineno}
\usepackage{thmtools}
\usepackage{enumitem}
\usepackage{comment}
\usepackage{etoolbox}
\usepackage{bbm}
\usepackage{subcaption}

\usepackage{scalerel}
\usepackage{tikz}
\usetikzlibrary{svg.path}

\definecolor{orcidlogocol}{HTML}{A6CE39}
\tikzset{
  orcidlogo/.pic={
    \fill[orcidlogocol] svg{M256,128c0,70.7-57.3,128-128,128C57.3,256,0,198.7,0,128C0,57.3,57.3,0,128,0C198.7,0,256,57.3,256,128z};
    \fill[white] svg{M86.3,186.2H70.9V79.1h15.4v48.4V186.2z}
                 svg{M108.9,79.1h41.6c39.6,0,57,28.3,57,53.6c0,27.5-21.5,53.6-56.8,53.6h-41.8V79.1z M124.3,172.4h24.5c34.9,0,42.9-26.5,42.9-39.7c0-21.5-13.7-39.7-43.7-39.7h-23.7V172.4z}
                 svg{M88.7,56.8c0,5.5-4.5,10.1-10.1,10.1c-5.6,0-10.1-4.6-10.1-10.1c0-5.6,4.5-10.1,10.1-10.1C84.2,46.7,88.7,51.3,88.7,56.8z};
  }
}

\newcommand\orcidicon[1]{\href{https://orcid.org/#1}{\mbox{\scalerel*{
\begin{tikzpicture}[yscale=-1,transform shape]
\pic{orcidlogo};
\end{tikzpicture}
}{|}}}}

\usepackage{hyperref}

\usepackage[linesnumbered,ruled,vlined]{algorithm2e}
\SetKw{Continue}{continue}

\usepackage{cleveref}

\title{The Hydrostructure: a Universal Framework for Safe and Complete Algorithms for Genome Assembly}

\author{Massimo Cairo \thanks{Department of Computer Science, University of Helsinki, Finland, \newline
Email:\{shahbaz.khan,sebastian.schmidt,alexandru.tomescu\}@helsinki.fi}
\and Shahbaz Khan \orcidicon{0000-0001-9352-0088} \footnotemark[1]~\thanks{This work was partially funded by the European Research Council (ERC) under the European Union's Horizon 2020 research and innovation programme (grant agreement No.~851093, SAFEBIO).}
\and Romeo Rizzi \orcidicon{0000-0002-2387-0952} \thanks{Department of Computer Science, University of Verona, Italy,
Email: romeo.rizzi@univr.it.}
\and Sebastian Schmidt \orcidicon{0000-0003-4878-2809} \footnotemark[1]~\footnotemark[2]
\and Alexandru~I.~Tomescu \orcidicon{0000-0002-5747-8350} \footnotemark[1]~\footnotemark[2]~\thanks{This work was partially funded by the Academy of Finland (grants No.~322595, 328877).} 
\and Elia~C.~Zirondelli \footnotemark[3]~\thanks{
Department of Mathematics, University of Trento, Italy, Email: eliacarlo.zirondelli@unitn.it.}}

\date{}

\graphicspath{{figures/}}

\newtheorem{theorem}{Theorem}[section]

\newtheorem{lemma}[theorem]{Lemma}
\newtheorem{definition}[theorem]{Definition}

\newtheorem{remark}[theorem]{Remark}

\newcommand{\sea}{\ensuremath{\mathsf{Sea}}}
\newcommand{\cloud}{\ensuremath{\mathsf{Cloud}}}
\newcommand{\river}{\ensuremath{\mathsf{River}}}
\newcommand{\vapor}{\ensuremath{\mathsf{Vapor}}}
\newcommand{\heart}{\ensuremath{\mathsf{Heart}}}

\newcommand{\cs}[1]{\ensuremath{\ifstrempty{#1}{k}{#1}}-circular safe}
\newcommand{\ls}[1]{\ensuremath{\ifstrempty{#1}{k}{#1}}-st safe}
\newcommand{\cvs}[1]{\ensuremath{\ifstrempty{#1}{k}{#1}}-circular \ensuremath{F}-visible safe}
\newcommand{\lvs}[1]{\ensuremath{\ifstrempty{#1}{k}{#1}}-st \ensuremath{F}-visible safe}
\newcommand{\ccs}[1]{\ensuremath{\ifstrempty{#1}{k}{#1}}-circular \ensuremath{F}-covering safe}
\newcommand{\lcs}[1]{\ensuremath{\ifstrempty{#1}{k}{#1}}-st \ensuremath{F}-covering safe}

\newcommand{\Rsib}[1]{\ensuremath{R_{sib}^+(\ifstrempty{#1}{e}{#1})}}

\newcommand{\head}[1]{\textsc{head}\ifstrempty{#1}{}{\ifstrequal{#1}{(}{(}{(#1)}}}
\newcommand{\tail}[1]{\textsc{tail}\ifstrempty{#1}{}{\ifstrequal{#1}{(}{(}{(#1)}}}
\newcommand{\wt}[1]{\textsc{start}\ifstrempty{#1}{}{\ifstrequal{#1}{(}{(}{(#1)}}}
\newcommand{\wh}[1]{\textsc{end}\ifstrempty{#1}{}{\ifstrequal{#1}{(}{(}{(#1)}}}

\begin{document}
\pagenumbering{gobble} 

\maketitle

\begin{abstract}

\emph{Genome assembly} is a fundamental problem in bioinformatics, requiring to reconstruct a source genome from an \emph{assembly graph} built from a set of \emph{reads} (short strings sequenced from the genome). Typically, in genome assembly, a solution is considered to be an arc-covering walk of the graph. Since assembly graphs admit many such solutions, a common approach is to find what is definitely present in all solutions, known as \emph{safe subsolutions}. For example, most practical assemblers use \emph{unitigs} at the core of their heuristics, namely paths whose internal nodes have unit in-degree and out-degree, and which are clearly safe. The long-standing open problem of finding \emph{all} the safe parts of the solutions was recently solved [RECOMB 2016] yielding a 60\% increase in contig length. This safe and \emph{complete} genome assembly algorithm was followed by other works improving the time bounds, as well as extending the results for different notions of assembly solution. However, all such results typically used very specific approaches, which did not generalise, and as a consequence could not handle practical issues. As such, it remained open whether one can be \emph{complete} also for models of genome assembly of practical applicability. Moreover, despite previous results presenting \emph{optimal} algorithms, they were based on avoiding {\em forbidden structures} (which are NO-certificates), and hence it was open whether there exists any simple YES-certificate to complete the understanding of the problem structure.

In this paper we present a \emph{universal framework} for obtaining safe and complete algorithms unifying the previous results, while also allowing for easy generalisations to incorporate many practical aspects. This framework is based on an entirely new perspective for studying safety, and on a novel graph structure (the \emph{hydrostructure}) highlighting the reachability properties of the graph from the perspective of a walk. The hydrostructure allows for simple characterisations of safe walks in existing and new models for genome assembly.
Moreover, the hydrostructure serves as a simple YES-certificate for {\em all} the studied models. We can directly adapt almost all our characterisations to \emph{optimal} verification algorithms, and \emph{simple} enumeration algorithms. Most of these enumeration algorithms are also improved to optimality using an incremental computation procedure and an existing optimal algorithm for the basic model~[ACM Trans.~Algorithms 2019, ICALP 2021].

On the theoretical side, we consider the hydrostructure as a generalisation of the standard notion of a \textit{cut}, giving a more flexible technique for studying safety of many other types of covering walks of a graph. On the practical side, we believe that the hydrostructure could prove to be a theoretical breakthrough for a fundamental bioinformatics problem, leading also to an improvement of \emph{practical} genome assembly from the point of view of \emph{completeness}.

\end{abstract}

\newpage
\pagenumbering{arabic} 
\setcounter{page}{1}

\section{Introduction}
\label{sec:intro}

Most problems in Bioinformatics are based at their core on some theoretical computational problem. After initial progress based on heuristics, several such Bioinformatics problems witnessed a drastic improvement in their practical solutions as a consequence of a breakthrough in their theoretical foundations. A major example of this is how the FM-index~\cite{FerraginaM00} revolutionised the problem of \emph{read mapping}, being central in tools such as~\cite{li2009soap2,li2009fast,langmead2012fast}. Other such theoretical breakthroughs include computing quartet distance~\cite{DudekG19,WilliamsWWY15} motivated by \emph{phylogenetics}, or fine-grained complexity lower bounds for \emph{edit distance} computation~\cite{BackursI15} motivated by \emph{biological sequence alignment}. However, despite this successful exchange of problems and results between Bioinformatics and Theoretical Computer Science, another flagship Bioinformatics problem, \emph{genome assembly}, is generally lacking similar developments, where {\em all} practical genome assemblers are fundamentally based on heuristics.
With our results we intend to fill exactly this gap
and enable a similar success story for practical genome assemblers, based on a theoretical result at their core.
Throughout this Introduction, we keep the Bioinformatics details at a minimum, and we refer the reader to \Cref{sec:bioinformatics-motivation} for more practical details and motivation.

\vspace{-.3cm}
\paragraph*{State-of-the-art in genome assembly.}
Given a collection of \emph{reads} (short strings sequenced from an unknown source genome), the task is to reconstruct the source genome from which the reads were sequenced. This is one of the oldest problems in Bioinformatics~\cite{peltola83}, whose formulations range from a shortest common superstring of the reads~\cite{peltola83,K92,KM95}, to various models of node- or arc-covering walks in different \emph{assembly graphs}, encoding the overlaps between the reads (such as \emph{de Bruijn graphs}~\cite{Pevzner1989}, or \emph{overlap graphs}~\cite{M05})~\cite{PTW01,MB09,MGMB07,kapun13b,SequencingHybridizationLysov1988,narzisi14,nagarajan2009parametric}. 
In general, most such theoretical formulations of different models of genome assembly are NP-hard~\cite{kingsford2010assembly,nagarajan2009parametric,narzisi14,kapun13b}. However, 
a more fundamental theoretical limitation is that such ``global'' problem formulations inherently admit a large number of solutions~\cite{kingsford2010assembly}, given the large size and complexity of the input data. 
In practice, genome assemblers output only shorter strings that are likely to be {\em correct} (i.e.~are substrings of the source genome)~\cite{nagarajan2013sequence,DBLP:books/cu/MBCT2015,medvedev2019modeling}. Such a strategy commonly uses the assembly graph to find only the paths whose internal nodes have \emph{unit} in- and out-degree (\emph{unitigs}). Since unitigs do not branch they are {\em correct}, and can be computed in linear time. The use of unitigs dates back to 1995~\cite{KM95} and is at the core of most state-of-the-art genome assemblers, for both {\em long reads} (such as wtdbg2~\cite{Ruan:2020aa}), and {\em short reads} (such as MEGAHIT~\cite{li2015megahit}). Even though {\em long reads} are theoretically preferable, due to various practical limitations {\em short reads} are still used in many biomedical applications, such as  the assembly of the SARS-CoV-2 genome~\cite{Wu:2020aa}.

\vspace{-.3cm}
\paragraph{Surpassing the theoretical limitations using safe and complete algorithms.} 
Despite being at the core of the state-of-the-art in both theory and practice, there is no reason why only unitigs should be the basis of correct partial answers to the genome assembly problem.
In fact, various works~\cite{journals/bioinformatics/Guenoche92,boisvert2010ray,nagarajan2009parametric,DBLP:journals/bioinformatics/ShomoronyKCT16,DBLP:journals/bmcbi/LamKT14,BBT13} presented the {\em open question} about the ``assembly limit'' (if any), or formally, {\em what is all that can be correctly assembled from the input reads}, by considering both graph theoretic and non-graph theoretic formulations. 
Unitigs were first generalised in~\cite{PTW01} by considering the paths having internal nodes with {\em unit} out-degree (with no restriction on in-degree). These were later generalised~\cite{MGMB07,jacksonthesis,kingsford2010assembly} evolving the idea of {\em correctness} to the paths of the form $P = P_1eP_2$, such that $e$ is an edge, the nodes of path $P_1$ have unit in-degree, and the nodes of path $P_2$ have unit out-degree (intuitively, $P_1$ is the only way to reach $e$, and $P_2$ is the only way $e$ reaches other nodes).
The question about the ``assembly limit'' was finally resolved in 2016 for a specific genome assembly formulation in a major Bioinformatics venue~\cite{tomescu2017safe} by introducing {\em safe and complete} algorithms for the problem, described as follows. 

Given a notion solution to a problem on a graph $G$ (i.e.~some type of walk in $G$), a partial solution (i.e.~a walk $W$) is called \emph{safe} if it appears in all solutions of the problem (i.e.~$W$ is a subwalk of all solutions on $G$). An algorithm reporting only safe partial solutions is called \emph{safe}. A safe algorithm reporting \emph{all} the safe partial solutions is called \emph{safe and complete}. A safe and complete algorithm outputs {\em exactly} the parts of the source genome that can be correctly reconstructed from the input reads. 
Notions similar to {\em safety} were studied earlier in 
Bioinformatics~\cite{Vingron01071990,Chao01081993,nagarajan2013sequence}, and in other fields, including {\em persistence}~\cite{doi:10.1137/0603052,Costa1994143,DBLP:journals/mmor/Cechlarova98}, \emph{$d$-transversals}~\cite{DBLP:journals/jco/CostaWP11}, \emph{$d$-blockers}~\cite{DBLP:journals/dm/ZenklusenRPWCB09}, and \emph{vital nodes/edges}~\cite{DBLP:journals/networks/BazganFNNS19}.

Given an assembly graph, the most basic notion of a solution (or a source genome) is that of a closed (i.e.~circular) walk covering all nodes or all arcs at least once, thereby explaining their existence in the assembly graph~\cite{tomescu2017safe,nagarajan2009parametric,narzisi14,medvedev2019modeling,nagarajan2013sequence,DBLP:books/cu/MBCT2015,kingsford2010assembly}.
The safe walks for this notion of solution include unitigs and their generalisations described above. Tomescu and Medvedev~\cite{tomescu2017safe} characterised the safe walks w.r.t.~closed arc-covering walks as {\em omnitigs}. On simulated error-free reads where the source genome is indeed a closed arc-covering walk, omnitigs were found to be on average 60\% longer than unitigs, and to contain 60\% more biological information without employing any heuristics.
Moreover,~\cite{tomescu2017safe} presented an $O(m^2n)$-time algorithm\footnote{Trivial analysis using new results about omnitigs proves $O(m^2n)$ time for~\cite{tomescu2017safe}, though not explicitly stated.} finding \emph{all} maximal safe walks for such genome assembly solutions, in a graph with $m$ arcs and $n$ nodes.
Later, Cairo et al.~\cite{DBLP:journals/talg/CairoMART19} improved this bound to $O(mn)$, which they also proved to be optimal by using graphs having $\Theta(mn)$-sized solutions. 
Recently, Cairo et al.~\cite{cairo2020macrotigs} presented a \emph{linear-time} output-sensitive algorithm for computing all maximal omnitigs, using a compact representation of the safe walks, called {\em macrotigs}.

\vspace{-.3cm}
\paragraph{Limitations of the existing theory of safe and complete algorithms.}
Despite presenting optimal algorithms, 
a theoretical limitation of the previous results from \cite{tomescu2017safe,DBLP:journals/talg/CairoMART19,acosta2018safe} is that safety is characterised in terms of \emph{forbidden structures} (i.e.~NO-certificates) for the safety of a walk.
These turned out to be an unnatural view on more advanced models of genome assembly, where only a subset of the omnitigs is safe (see the next subsection).
As such, these characterisations were incomplete and unnatural in the absence of easily verifiable YES-certificates.
To illustrate this, consider for example the classical notion of a \emph{strong bridge} $(x, y)$ (see \Cref{fig:bridge-cut}).
Its NO-certificate is a path from $x$ to $y$ avoiding $(x,y)$, which can thus be seen as a forbidden path. 
The corresponding YES-certificate is a \emph{cut} between a set of nodes $S$ (containing at least the nodes reachable from $x$ without using $(x,y)$), and the remaining nodes $T$ (containing at least the nodes reaching $y$ without using $(x,y)$), such that the only arc crossing the cut is $(x,y)$.
Such a certificate captures much more information about the structure of the graph from the viewpoint of the arc $(x, y)$.
We generalise this idea from single arcs to walks to get a new perspective on safety problems in genome assembly. For that, we use a similar graph structure (now recognizing \emph{safe} walks) which is essentially a generalisation of a cut, and hence a YES-certificate, which is described as follows.

\begin{figure}[h]
    \centering
    \includegraphics[scale=0.8]{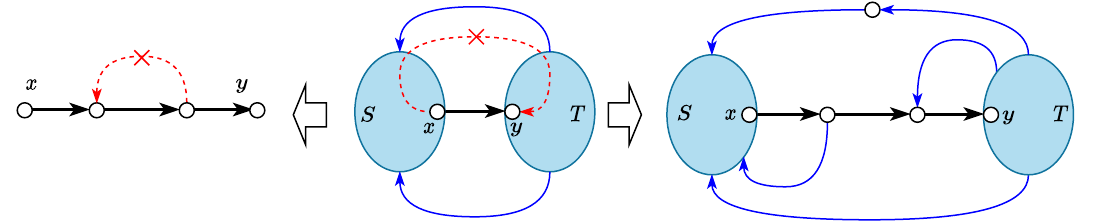}
    \caption{
    \small The perspective of a safe walk (bold black), generalizing a strong bridge. 
    Center: The strong bridge $(x,y)$ has a NO-certificate in form of a forbidden path (red), and a YES-certificate in form of a cut between $S$ and $T$ where $(x,y)$ is the only edge leaving $S$. Left: The NO-certificate of the safety of the walk from $x$ to $y$ is a similar forbidden path (red). Right: The YES-certificate of the safety of the walk from $x$ to $y$ is a graph partition 
    where $T$ is reachable from $S$ only by using the walk from $x$ to $y$ contiguously. 
    }
    \label{fig:bridge-cut}
\end{figure}

Our perspective to study safety distinguishes between the trivial reason for safety, i.e., the {\em covering} constraint of an arc (and the left and right extensions the arc), from the more profound reason arising from the {\em bridge-like} nature of some walks, which generalises the property of a strong bridge.
Such {\em bridge-like} walks are required to be traversed contiguously for reachability between some nodes. Thus, analogous to NO-certificates of a strong bridge (see \Cref{fig:bridge-cut}), 
a bridge-like walk (say from $x$ to $y$) requires the absence of a forbidden path which allows reaching $y$ from $x$ without traversing the whole walk contiguously, as described in previous results~\cite{tomescu2017safe,DBLP:journals/talg/CairoMART19,acosta2018safe,cairo2020macrotigs}. Similarly, analogous to the YES-certificate of a strong bridge (see \Cref{fig:bridge-cut}), a bridge-like walk has a directed {\em cut-like} structure between the set of nodes $S$ reachable from $x$ without traversing the whole walk, and the set of nodes $T$ reaching $y$ without traversing the whole walk. Crossing the cut from $y$ to $x$ uses the remainder of the graph, which are the nodes and edges outside the induced subgraphs of $S$ and $T$. On the other hand, crossing the cut from $x$ to $y$ (or $S$ to $T$) requires traversing the whole walk contiguously.
This partitions the whole graph from the perspective of the walk, reducing the requirement for the contiguous traversal of the walk, to the simple requirement of reaching $y$ from $x$, allowing simpler characterisations of more advanced models.
Moreover, we can now use the same YES-certificate for many different models, as opposed to finding a NO-certificate separately for each model.
Thus, our new perspective of a safe walk as a generalisation of a strong bridge
results in a universal approach for the complete characterisation of safe walks.

\vspace{-.3cm}
\paragraph{Formulation of practically relevant genome assembly models.}
Modeling a genome assembly solution as a single arc-covering walk is extremely limiting in practice, due to the  presence of multiple genomes in the sample (not necessarily circular), sequencing errors or unsequenced genomic regions. See \Cref{sec:bioinformatics-motivation} for further motivation of our definitions below, and a discussion about how to handle practical aspects of the data such as multiplicities (i.e. k-mer abundance).

Most of these practical issues can be handled by considering more flexible theoretical formulations of the problem. Instead of always considering the solution to be a \emph{single} closed arc-covering walk of the assembly graph, we can change the model so that the solution is an \emph{arc-covering} collection of up to $k \geq 2$ closed walks (i.e.~every arc appears in some walk of the collection). 
Further, for addressing \emph{linear} genomes, or unsequenced genomic regions, we further change the model so that the solution is one \emph{open} arc-covering walk from a given node~$s$ to a given node~$t$ (\emph{$s$-$t$ walk}), or an arc-covering collection of up to $k \geq 2$ open $s$-$t$ walks. Moreover, if there is no constraint $k$ on the number of walks in the collection, then we will say that $k = \infty$.

\begin{definition}[\cs{k} walk, \ls{k} walk]
\label{def:safe-walks-circular-linear}
Let $G = V \cup E$ be a graph, let $s,t \in V$ and let $k \geq 1$. A walk $W$ is called \emph{\cs{k}} (or \emph{\ls{k}}) if $W$ is a subwalk of at least one walk of any arc-covering collection of up to $k$ circular walks (or up to $k$ walks from $s$ to $t$).
\end{definition}

\begin{remark}
    A graph admits an arc-covering collection of $k \geq 1$ circular walks if and only if it is a disjoint union of at most $k$ strongly connected graphs. As such, in the circular models we can assume the graph to be strongly connected. In the linear models, we first solve the strongly connected case, and then solve the cases $k = 1,\infty$ for non-strongly connected graphs.
\end{remark}

Further, the notion of genome assembly solution can be naturally extended to handle sequencing errors. For example, the models can be extended so that the collection of walks is required to cover only a \emph{subset} $F$ of the arcs (\emph{$F$-covering}). Another possible extension is to mark the erroneous arcs in $E \setminus F$ as \emph{invisible}, in the sense that they are invisible when we define safety (\emph{$F$-visible}, see below). These not only allow handling errors, but also allow handling even more general notions of genome assembly solutions through simple reductions (see \Cref{rem:subset-covering-visibility-applications}).

\begin{definition}[Subset covering / visible]
\label{def:safe-walk-covering-visible}
Let $G = V \cup E$ be a graph, let $s,t \in V$, $F\subseteq E$, and let $k \geq 1$. A walk $W$ is called:
\begin{itemize}[nolistsep]
    \item \emph{\ccs{k}} (or \emph{\lcs{k}}) if $W$ is a subwalk of at least one walk of every $F$-covering collection of up to $k$ circular walks (or up to $k$ walks from $s$ to $t$).
    \item \emph{\cvs{k}} (or \emph{\lvs{k}}) if the $F$-subsequence of $W$ occurs contiguously in the $F$-subsequence of at least one walk of every arc-covering collection of up to $k$ circular walks (or up to $k$ walks from $s$ to $t$). The \emph{$F$-subsequence} of a walk is obtained by removing its arcs \emph{not in $F$}.
\end{itemize}
\end{definition}

\begin{remark}
\label{rem:subset-covering-visibility-applications}
The subset covering model can also solve a generalisation of the linear models, 
where the walks in the collection can start and end in any node in given \emph{sets} $S$ and $T$, respectively:
set $F = E$, and add new nodes $s$ and $t$ connected to / from all nodes in $S$ and $T$, respectively. Moreover, we can also combine the subset covering and visibility models for some $F,F'\subseteq E$, to get \emph{$F$-covering $F'$-visible safe walks}, for both circular and linear models, and obtain analogous results (see also \Cref{fig:characterisation-overview}). By this, we can also solve the same models in a node-centric formulation, where only (a subset of) the nodes are required to be covered and/or visible: expand each node to an arc and choose only such node-arcs to be covered and/or visible.
\end{remark}

\section{Contributions}
We obtain safe and complete algorithms for a plethora of \emph{natural} genome assembly formulations, as stated in \Cref{def:safe-walks-circular-linear,def:safe-walk-covering-visible} and \Cref{rem:subset-covering-visibility-applications}, with an entirely {\em new perspective}, and a \emph{universal framework} for safety, the \emph{hydrostructure} of a walk.

\vspace{-.3cm}
\paragraph{The Hydrostructure as a Universal Framework for safety.} 
The hydrostructure gives a more structured view of the graph from the perspective of the walk, allowing for simple safety characterisation for \emph{all} models (see \Cref{fig:characterisation-overview}).
We characterise \cs{k} walks for any given $k$ (where $1<k<\infty$), which was mentioned as an open problem in~\cite{acosta2018safe}, and prove the equivalence of the \cs{k} problems for all $k\geq 2$.
We also present the \emph{first} characterisations for the {\em single and multiple linear} models (mentioned as open problems in~\cite{tomescu2017safe}), which are in general harder because the location of $s$ and $t$ affects safety. 
Moreover, the {\em flexibility} of the hydrostructure allows us to generalise these to the more practically relevant {\em subset covering} and {\em visibility} models.

Even though the hydrostructure is developed mainly for such arc-covering models, it reveals a novel fundamental insight into the structure of a graph, which may be of independent interest (see \Cref{sec:hydrostructure}). It captures the bridge-like characteristic of a walk for characterizing safety. This distinguishes the core of a safe walk from its trivially safe extensions on both sides, i.e., along the only way to reach it (its {\em left wing}, if exists), and along the only way forward from it (its {\em right wing}, if exists).
Its water-inspired terminology stems from the standard notions of ``source'' and ``sink'' nodes, but in the present case of a strongly connected graph now mimics the water cycle on Earth~\cite{wiki:waterCycle}. The hydrostructure distinguishes the source part as $\sea$ and the target part as $\cloud$. The internal part of the bridge-like walk forms the $\vapor$, which is the only way water moves from $\sea$ to $\cloud$. On the other hand, water can easily move from $\cloud$ to $\vapor$ and from $\vapor$ to $\sea$. The rest of graph forms the $\river$, which in general serves as an alternate %
route for the water from $\cloud$ to $\sea$. For non-bridge-like walks the hydrostructure trivially reduces to $\vapor$ being the entire graph, and the remaining components being empty. For any walk, we show that its hydrostructure can be easily computed in {\em linear} time, by evaluating the restricted reachability of the end points of the walk.

\begin{figure}
    \centering
    \includegraphics[trim={0 0 0 10},clip,scale=.75]{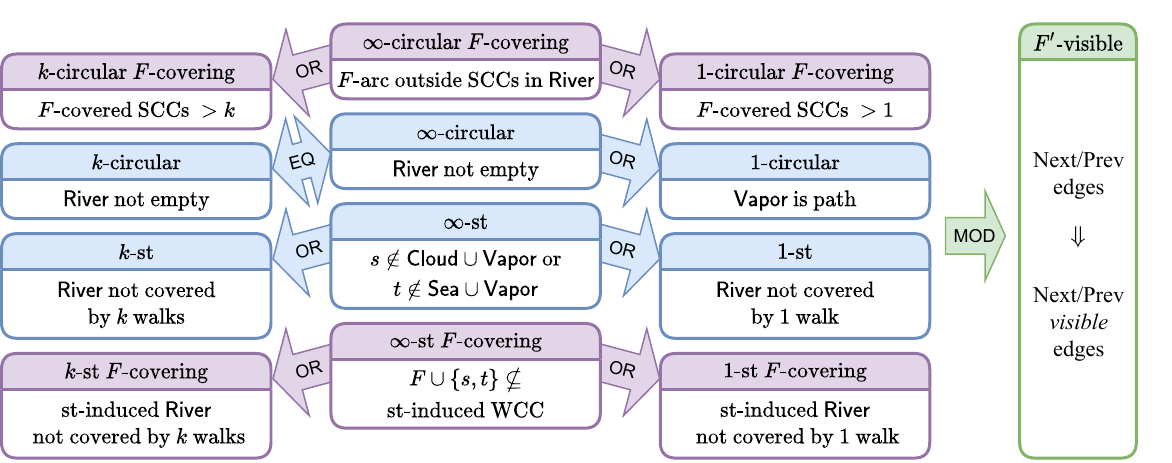}
    \caption{\small %
    An overview of the characterisations of safety for non-trivial walks in all our models. The arrow labeled EQ means equivalent criteria; an arrow labeled OR means that the target characterization is the OR between the criteria of the endpoints of the arrow; the arrow labeled MOD means modification of the source characterization as described.
    For simplicity of presentation, the criteria for trivial walks and wings (for linear models) are not included, despite being simple criteria based on hydrostructure.
    Note that all characterisations imply that the \vapor{} is a path, (otherwise \sea{}, \cloud{} and \river{} would be empty, making all the criteria trivially false) and hence include being \cs{1}.
    }
    \label{fig:characterisation-overview}
\end{figure}

\vspace{-.3cm}
\paragraph{Simple and \emph{full} characterisation for existing and new problems.}
For single circular walks, the previous characterisations identified forbidden paths \cite{tomescu2017safe,DBLP:journals/talg/CairoMART19} resulting only in a NO-certificate. Our characterisation simplifies it to merely $\vapor$ not being the entire graph, which basically signifies that the walk is bridge-like, resulting in easily verifiable both YES- and NO-certificates. For $\infty$-circular safe walks, the previous characterisation further requires the presence of a {\em certificate arc}~\cite{acosta2018safe} forbidding cycles of certain types, resulting in a NO-certificate which is even harder to verify than in the 1-circular model. Our characterisation only adds an additional constraint of having a non-empty $\river$, drastically simplifying the characterisation. Note that it is the same for any given number $k \geq 2$ of closed walks, proving their equivalence. See \Cref{s:closed} for details.

For linear safe models, surprisingly every solution has to be \cs{1}, allowing us to build the characterisation (and hence algorithms) of linear models on top of those for the single circular model. In the single linear model, $\river$ additionally has to be a single path, with some criteria based on the location of $s$~and~$t$. However, the multiple linear model generalises the criteria on the $\river$ to be based on its {\em minimum walk-cover} being of a certain size. Also, we require additional characterisations for the safety of the wings, based on the reachability of $s$ and $t$ in the candidate walk. See \Cref{s:open} for details.
In the subset covering model, the previous characterisations are easily extendable using more concrete criteria involving the corresponding subset $F$. In the subset visibility model, the hydrostructure itself is generalised for the corresponding subset $F$, so that the existing characterisations can be directly applied without any changes. Finally, for all these problems, the hydrostructure itself (or in some cases with additional walk cover\footnote{For simplicity we use the walk cover (a NO-certificate) in our characterisation, which have a simple equivalent YES-certificate using \textit{maximum arc antichain.}}) serves as both the YES- and NO-certificate. See \Cref{sec:subsetCov} for details.

The simplicity of our characterisations motivates us to introduce the problem of {\em verifying} whether a given walk is safe in a model. Despite this problem not being explicitly studied earlier, the previous characterisations~\cite{tomescu2017safe,DBLP:journals/talg/CairoMART19,acosta2018safe} resulted in $O(mn)$ time verification algorithms, which again depend on the corresponding certificate for the model. Recently, Cairo et al.~\cite{cairo2020macrotigs} presented an algorithm which can be adapted to a linear-time verification algorithm for \cs{1} walks, but it uses complex data structures. 
Our characterisations are directly adaptable to linear time {\em optimal} verification algorithms for almost all of our models, using simple techniques such as graph traversal. See \Cref{tab:results} for a comparison.

\begin{table}[t!]
\small 
    \centering
    \begin{tabular}{|c|c|c|c|c|c|}
\hline        
\textbf{Safety} &  \multicolumn{2}{c|}{\textbf{Previous Results}} & \multicolumn{3}{c|}{\textbf{New Results}} \\
\cline{2-6}
\textbf{Problems} & \textbf{Verify} & \textbf{Enumerate} & \textbf{Verify} & \textbf{Enum. Trivial} & \textbf{Enum. Improved}\\ \hline
$1$-circular & $O(mn)$ &  $O(m^2n)$~\cite{tomescu2017safe} & $O(m)^*$ & $O(m^2n)$ & - \\
 & $O(mn)$ & $O(mn)^*$~\cite{DBLP:journals/talg/CairoMART19} &  & & \\
& $O(m)^*$ & $O(m+o)^*$~\cite{cairo2020macrotigs} & & & \\ 
$k$-circular & - & - & $O(m)^*$ & $O(m^2n)$& $O(mn)^*$ \\
$\infty$-circular & $O(mn)$ & $O(m^2n)$~\cite{acosta2018safe} & $O(m)^*$ & $O(m^2n)$ & $O(mn)^*$ \\
\hline
$1$-st & - & - & $O(m)^*$ & $O(m^2n)$  & $O(mn+o)^*$ \\
$k$-st & - & - & $O(mn)^\#$ & $O(m^2n^2)^\#$   & $O(m^2n+o)^\#$ \\
$\infty$-st & - &  - & $O(m)^*$ & $O(m^2n)$ & $O(mn+o)^*$ \\
\hline
    \end{tabular}
    \caption[Our runtime bounds.]{\small A comparison of the previous results for safety problems with the new results trivially obtained from the characterisation, and the improved results. The size of the output is represented by $o$.
    The optimal algorithms are marked by~$^*$.
    The time bounds using a simple minimum walk cover algorithm are marked by~$^\#$, which may be improved\footnotemark{} with faster algorithms for computing the size of a minimum walk cover. The solution for this model ($^\#$) is also limited to {\em strongly connected graphs}.
    All the above models can also be extended to {\em subset covering} using the same bounds and to {\em subset visibility} with an additive $O(mn)$ term.
    }
    \label{tab:results}
\end{table}

\footnotetext{
    Our solutions for \ls{k} for {\em general} $k$ requires the query ``is the size of a minimum walk-cover of a subgraph greater than $k$?''. This can be answered using flows in $f=O(mn)$ time, and maintained during the incremental computation in $g=O(m^2)$ time, resulting in the stated bounds. Considering them as black box algorithms, our bounds for verification, enumeration and improved enumeration algorithm are $O(m+f)$, $O(m^2n+mnf)$ and $O(mn+o+ng)$ respectively.}

\vspace{-.3cm}
\paragraph{Novel techniques to develop optimal enumeration algorithms.}
Our verification algorithms can also be adapted to simple $O(m^2n)$ time algorithms (except for {\em given} \ls{k}) to enumerate all the maximal safe walks. Further, using the optimal \cs{1} algorithm~\cite{DBLP:journals/talg/CairoMART19} and our new characterisation, we not only improve upon \cite{acosta2018safe} for $\infty$-circular walks, but also make our \cs{k} algorithm {\em optimal}.
In order to improve our linear algorithms we present a {\em novel technique} to compute the hydrostructure for all subwalks of a {\em special} walk using incremental computation. Since the linear characterisations are built on top of the \cs{1} walks, we use the concise representation of \cs{1} walks in $O(n)$ special walks~\cite{cairo2020macrotigs}.
At its core, the incremental computation uses the incremental reachability  algorithm~\cite{Italiano86} requiring total $O(m)$ for each special walk. Surprisingly, every node and arc can enter and leave $\river$ exactly once during the incremental computation, allowing its maintenance in total $O(m)$ time. 
This results in {\em optimal} \ls{k} algorithms for $k= 1,\infty$, while for the general $k$ we also require to check if a subgraph can be covered with $k$ walks. Moreover, we prove the optimality of our results (except for \emph{given} $k$) by presenting a family of graphs having the size of all maximal safe walks $\Omega(mn)$. See \Cref{tab:results} for a comparison, and \Cref{sec:algorithms} for further details.

\medskip
Summarising, the hydrostructure is a {\em mathematical tool} which is {\em flexible} for addressing a variety of models (handling different types of genomes, errors, complex or unsequenced regions), and is {\em adaptable} to develop simple yet efficient algorithms. Thus, we believe the various results of this paper can form the theoretical basis of future \emph{complete} genome assemblers. 
We hope that our paper together with~\cite{cairo2020macrotigs} 
can enable the state-of-the-art genomic tools to have advanced theoretical results at their core, as e.g.~the FM-index in read mapping tools, leading to similar success story for genome assembly.

\section{Preliminaries}
\label{sec:preliminaries}

\paragraph{Graph Notations.}
A \emph{graph} is a set $G = V \cup E$, where $V$ is a finite set of \emph{nodes}, $E$ is a finite multi-set of ordered pairs of nodes called \emph{arcs}; given an arc $e = (u,v)$, we say that $u$ is the \emph{tail} node $\tail{}(e)$ of $e$, and $v$ is the \emph{head} node $\head{}(e)$ of $e$.
Self-loops are allowed. 
A \emph{strongly connected component} (SCC) is a maximal subgraph that is strongly connected, i.e. for any two nodes $x$ and $y$ in the SCC, there exists a path from $x$ to $y$. Similarly, a {\em weakly connected component} (WCC) is a maximal subgraph that is weakly connected, i.e., any two nodes $x$ and $y$ in the WCC are connected by an undirected path.
A node is a \emph{source} if it has no incoming arcs, and a \emph{sink} if it has no outgoing arcs.
For $e \in E$, we denote $G \setminus e = V \cup E \setminus \{e\}$.
In the rest of this paper, we assume a fixed strongly connected graph $G= V \cup E$ which is not a cycle~\footnote{Safe walks in a cycle are not properly defined as they can repeat indefinitely.}, with $|V| = n$ and $|E| = m > n$.

A \emph{$w_1$-$w_\ell$ walk} (or simply \emph{walk}) in $G$ is a non-empty alternating sequence of nodes and arcs $W = (w_1, \dots, w_\ell)$, where for all $1 \leq i < \ell$: $\head{}(w_{i}) = w_{i+1}$ if $w_{i+1}$ is a node, and $\tail{}(w_{i+1}) = w_i$ otherwise.
Deviating from standard terminology, if not otherwise indicated, $w_1$ and $w_\ell$ are \emph{arcs}.
We call $\wt{}(W) = w_1$ the \emph{start} of $W$, and $\wh{}(W) = w_\ell$ the \emph{end} of $W$.
$W$ is a \emph{path} if it repeats no node or arc, except that $\wt{}(W)$ may equal $\wh{}(W)$.

A walk $W$ is called \emph{closed} if $\wt{}(W) = \wh{}(W)$ and $\ell > 1$, otherwise it is \emph{open}.
The notation $WW'$ denotes the concatenation of walks $W = (w_1, \dots, w_\ell)$ and $W' = (w'_1, \dots, w'_{\ell'})$, if $(w_1, \dots, w_\ell, w'_1, \dots, w'_{\ell'})$ is a walk. When writing a walk as a concatenation, then lower-case letters denote single nodes or arcs, i.e. $aZb$ denotes the walk $(a, z_1, \dots, z_\ell, b)$ (where $a$ and $b$ are arcs).
Subwalks of walks are defined in the standard manner, where subwalks of closed walks do not repeat the start/end if they run over the end.
A walk $W$ that contains at least one element is \emph{bridge-like} if there exist $x, y \in G$ such that each $x$-$y$ walk contains $W$ as subwalk, and otherwise it is called \emph{avertible}.
Observe that a bridge-like walk is an open path.

A \emph{walk-cover} of a graph is a set of walks that together cover all arcs of the graph, and its size is the number of walks.
The minimum walk-cover (and hence its {\em size}) can be computed in $O(mn)$ time using minimum flows, by reducing the problem to maximum flows (see e.g. \cite{bang2008digraphs}) and applying Orlin's and King's $O(mn)$ time algorithms \cite{orlin2013max,king1994faster}. %

\vspace{-.3cm}
\paragraph{Safety Notations.}
We call a node $v$ a \emph{join node} (or \emph{split node}) if its in-degree (or out-degree) is greater than $1$.
Similarly, an arc $e$ is a \emph{join arc} (or \emph{split arc}) if $\head{e}$ is a join node (or $\tail{e}$ is a split node).
Two arcs $e$ and $e'$ are \emph{sibling arcs} if $\head{}(e) = \head{}(e')$ or $\tail{}(e) = \tail{}(e')$.
A walk $W = (w_1, \dots, w_\ell)$ is \emph{univocal} if at most $w_\ell$ is a split node (if it is a node), and it is \emph{R-univocal} if at most $w_1$ is a join node (if it is a node), and it is \emph{biunivocal} if it is both univocal and R-univocal.
Its \emph{univocal extension} $U(W)$ is $W^lWW^r$ where $W^l\wt(W)$ is the longest R-univocal walk to $\wt(W)$ and $\wt(W)W^r$ is the longest univocal walk from $\wh(W)$.

\begin{figure}[htb]
    \centering
    \includegraphics[scale=1.4]{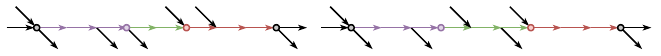}
    \caption{\small A trivial walk on the left and a non-trivial walk on the right.
    The heart is colored green, the left wing violet and the right wing red.
    The left wing is from arc to node and the right wing from node to arc.}
    \label{fig:hearts_and_wings}
\end{figure}

For a walk $W = w_1, \dots, w_\ell$ let $w_i$ be its first join arc or $w_i = w_1$ if $W$ has no join arc, and let $w_j$ be its last split arc or $w_j = w_\ell$ if $W$ has no split arc.
If $i \geq j$, then the \emph{trivial heart} (or simply \emph{heart}) $\heart(W)$ of $W$ is its $w_j$-$w_i$ subwalk, and otherwise the \emph{non-trivial heart} (or simply \emph{heart}) $\heart(W)$ is its $w_i$-$w_j$ subwalk.
A walk with a trivial heart is a \emph{trivial walk}, and a walk with a non-trivial heart is a \emph{non-trivial walk}.
The \emph{left wing} and \emph{right wing} of $W$ are $W^l$ and $W^r$ in the decomposition $W^l\heart(W)W^r$.
See \Cref{fig:hearts_and_wings} for a visualisation of these definitions.

Further, we use the following results from previous works in the literature. The first result was used in~\cite{acosta2018safe,cairo2020macrotigs} (and possibly other previous works) even though not stated explicitly.

\begin{theorem}[restate = prelim, name = ]
    \label{thm:prelim}
    For a strongly connected graph $G$ with $n$ nodes and $m$ arcs, the following hold:
    \begin{enumerate}[label = (\alph*), nosep]
        \item \label{thm:two-pointer-algirithm} \emph{(Two-Pointer Algorithm \cite{acosta2018safe,cairo2020macrotigs})} Given a walk $W$ and a procedure $A$ to verify the safety of its subwalks where all subwalks of a safe walk are safe, then all the maximal safe subwalks of $W$ can be reported in $O(|W|f(m,n))$ time, where each invocation of $A$ requires $f(m,n)$ time.
        \item \label{item:omnitig-properties} \emph{(Omnitig Bounds~\cite{DBLP:journals/talg/CairoMART19,cairo2020macrotigs})} There are at most $m$ maximal \cs{1} walks~\cite{DBLP:journals/talg/CairoMART19} where each has length of at most $O(n)$~\cite{DBLP:journals/talg/CairoMART19} and of which $O(n)$ are non-trivial~\cite{cairo2020macrotigs}, and the total length of all maximal \cs{1} walks in $G$ is $O(mn)$~\cite{DBLP:journals/talg/CairoMART19} (if $G$ is not a cycle).
        \item \label{thm:prelim:georgiadis} {\em (Fault tolerant SCCs \cite{georgiadis2017strong})} We can preprocess $G$ in $O(m)$ time, to report whether $x$ and $y$ are in the same SCC in $G \setminus z$ in $O(1)$-time, for any $x, y, z \in G$. %

    \end{enumerate}
\end{theorem}

\section{Hydrostructure}
\label{sec:hydrostructure}

The hydrostructure of a walk partitions the strongly connected graph from the perspective of the walk; for a \emph{bridge-like} walk it identifies two parts of the graph separated by the walk, and gives a clear picture of the reachability among the remaining parts. This allows an easy characterisation in problems (such as safety) that inherently rely on reachability. 
It is defined  using the {\em restricted forward and backward reachability} for a walk, as follows.

\begin{definition}
    The \emph{restricted forward and backward reachability} of a walk $W$ is defined as
    \begin{itemize}[nosep]
        \item $R^+(W) = \{x \in G \mid \exists \text{ $\wt{}(W)$-$x$ walk } W' \text{s.t.~} W \text{ is not a subwalk of $W'$}\}$,
        \item $R^-(W) = \{x \in G \mid \exists \text{ $x$-$\wh{}(W)$ walk } W' \text{s.t.~} W \text{ is not a subwalk of $W'$}\}$.
        \end{itemize}
\end{definition}

By definition, all nodes and arcs of $W$ (except for its last arc $\wh{}(W)$, since by our convention the walks start and end in an {\em arc}), are always in $R^+(W)$ (and symmetrically for $R^-(W)$). Further, if a walk is not bridge-like, we additionally have $\wh{}(W) \in R^+(W)$ as we will shortly prove, and $R^+(W)$ extends to the entire graph $G$ due to strong connectivity. Thus, inspired by the similarity between the water cycle on Earth and the reachability among the four parts of the Venn diagram of $R^+(W)$ and $R^-(W)$ (see \Cref{fig:Hydro-venDiagram}), we define the hydrostructure of a walk as follows.

\begin{definition}[Hydrostructure of a Walk]
 \label{def:clouds}
 Let $W$ be a walk with at least two arcs.
 \begin{itemize}[nosep]
    \item
     $\sea(W) = R^+(W) \setminus R^-(W)$.
   \item
     $\cloud(W) = R^-(W) \setminus R^+(W)$.
   \item
     $\vapor(W) = R^+(W) \cap R^-(W)$.
   \item
     $\river(W) = G \setminus (R^+(W) \cup R^-(W))$.
 \end{itemize}
 \end{definition}

\begin{figure}[htb]
    \centering
    \begin{subfigure}[m]{0.48\linewidth}
    \centering
    \includegraphics[trim=0 2 10 1, clip, scale=.85]{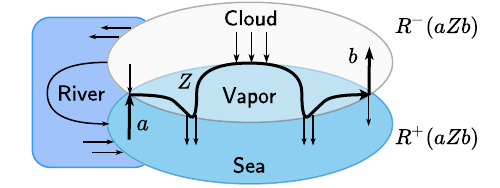}
    \caption{\label{fig:Hydro-venDiagram}}
    \end{subfigure}
    \hspace{1em}
    \begin{subfigure}[m]{0.48\linewidth}
    \centering
    \includegraphics[trim=0 2 0 1, clip, scale=0.85]{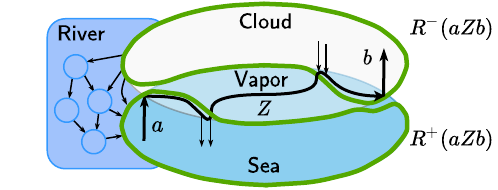}
    \caption{\label{fig:Hydro-properties}}
    \end{subfigure}
    \caption{\small (a): The hydrostructure partitioning of a non-trivial heart $aZb$ in the sets $\sea$, $\cloud$, $\vapor$, and $\river$ using $R^+(aZb)$ and $R^-(aZb)$. It also shows the membership of the arcs, and the reachability among the parts. (b): The hydrostructure for a trivial bridge-like walk $aZb$, with the sea-related and cloud-related SCC shown using {\em thick green} cycles.
    The $\river$ is a DAG of SCCs where the SCCs are blue cycles.
    }
    \label{fig:clouds}
\end{figure}

The reachability among the parts of the hydrostructure mimics the water cycle, which we describe as follows (formally proved later).
For any node or arc in the $\sea$, the only way to reach the $\cloud$ is by traversing the walk through the entire $\vapor$ (similar to the evaporation process), as it is shared by both $R^+(W)$ and $R^-(W)$.
Thus, for a bridge-like walk $W$, the hydrostructure exposes the $\sea$ and the $\cloud$ that are separated by $W$ in the $\vapor$.

On the other hand, a node or an arc from the $\cloud$ can directly reach the $\vapor$ (by dissipation), being in $R^-(W)$, and from the $\vapor$ it can directly reach the $\sea$ (by condensation), being in $R^+(W)$.
Finally, the $\river$ acts as an alternative path from the nodes and arcs in the $\cloud$ to the $\sea$ (by rainfall), since the forward reachability from the $\cloud$ and the backward reachability from the $\sea$ are not explored in $R^+(W)$ and $R^-(W)$.

\paragraph{Properties.}

The strong connectivity of the graph results in the hydrostructure of an {\em avertible} walk to have the entire graph as the $\vapor$, whereas the $\sea$, $\cloud$, and $\river$ are empty. On the other hand, for a bridge-like walk, the $\vapor$ is exactly the internal path of the walk, resulting in the following important property for any walk $W=aZb$. %

\begin{lemma}[restate = cloudcases, name = ]
    \label{lem:cloud-cases}
    For a walk $W$, $\vapor(W)$ is the open path $Z$ iff $W$ is bridge-like, otherwise $\vapor(W)$ is $G$.
\end{lemma}

Note that the hydrostructure partitions not just the nodes but also the arcs of the graph. The strong connectivity of the graph and the definitions above result in the following properties of the parts of the hydrostructure (see \Cref{fig:Hydro-properties} for the SCCs, and \Cref{fig:Hydro-venDiagram} for reachability).

\begin{lemma}[restate = clouds, name = ]
    \label{lem:clouds}
    The hydrostructure on a bridge-like path $aZb$ exhibits the following properties:
    \begin{enumerate}[label = \alph*., ref = (\alph*), nosep]
        \item \emph{(Safety)} \label{lem:clouds:forwards}
        Every walk from $x \in \sea(aZb)$ to $y \in \cloud(aZb)$ contains $aZb$ as a subwalk.
        
        \item \emph{(Avoidance)} \label{lem:clouds:backwards}
        Every $x \in \cloud(aZb)$ and $y \in \sea(aZb)$ have an $x$-$y$ walk with no subwalk $aZb$.
        
        \item \emph{(Separation)} \label{lem:clouds:separation}
        $a \in \sea(aZb)$ and $b \in \cloud(aZb)$.
        
        \item \emph{(SCCs)} \label{lem:clouds:connectedseacloud}
        If $\sea(aZb)$ is not just $\{a\}$,
        then $\sea(aZb)$ and the prefix of $aZb$ ending in its last split-node form an SCC (denoted as \emph{sea-related SCC}).
        Similarly, if $\cloud(aZb)$ is not just $\{b\}$, then $\cloud(aZb)$ and the suffix of $aZb$ starting in its first join-node form an SCC (denoted as \emph{cloud-related SCC}).
        
        \item \emph{(Reachability)} \label{lem:clouds:connectivity}
        Let $(x, y) \in G\times G$ be an ordered incidence pair (arc-node  or node-arc) in $G$, where $x$ and $y$ are in different components of the hydrostructure of $aZb$.
        Then $(x, y)$ can (only) be associated with the following pairs of components.
        \begin{itemize}[nosep]
           \item  $(\vapor, \cloud)$ always where $y = b$, or $(\sea, \vapor)$ always where $x = a$,
           \item $(\cloud, \vapor)$ or $(\vapor, \sea)$, where both always occur for non-trivial $aZb$,
           \item $(\cloud, \river)$ or $(\river, \sea)$, where both always occur for nonempty $\river$,
           \item $(\cloud, \sea)$, where $aZb$ is univocal or R-univocal. 
        \end{itemize}
    \end{enumerate}
\end{lemma}

\paragraph{Implementation.}
To obtain the hydrostructure of a walk $W$, we need to compute $R^+(W)$ (and symmetrically $R^-(W)$).
Note that $R^+(W)$ for avertible walks is simply $G$, and for bridge-like walks is the graph traversal from $\wt(W)$ in $G\setminus \wh(W)$.
We can identify if $W$ is avertible while performing the traversal, by checking if the traversal reaches $W$ again after leaving it. Thereafter, we can easily compute the $\sea$, $\cloud$, $\vapor$ and $\river$ by processing each arc and node individually (see \Cref{apx:implementation:hydrostructure} for more details).

\begin{theorem}[restate = hydrostructurebasicalgo, name = ]
\label{thm:hydro_algo}
    The hydrostructure of any walk can be computed in $O(m)$ time.
\end{theorem}

\section{Safety in Circular Models}
\label{s:closed}

In the circular models the \emph{heart} of a walk determines its safety, as the univocal extension of a safe walk is naturally safe.
Therefore, every trivial walk is naturally safe because of the covering constraint for any arc in its heart, which can be univocally extended to get the whole walk.

A non-trivial walk $W$ with a bridge-like heart is safe as well, because every circular arc-covering walk contains all bridge-like walks by definition.
If on the other hand $\heart(W)$ is avertible (having $\vapor(W) = G$), then a walk can cover the entire graph without having $\heart(W)$ as subwalk, making it unsafe by definition.
Thus, the characterisation for circular safe walks is succinctly described as follows (see \Cref{apx:closed} for the proofs of this section):

\begin{restatable}[1-Circular]{theorem}{onecircsafe}
\label{th:circularSafe}
A non-trivial walk $W$ is 1-circular safe iff $\vapor(\heart(W))$ is a path.
\end{restatable}

Now, for $k \geq 2$, a single circular arc-covering walk can be repeated to get $k$ circular arc-covering walks, implying that every \cs{k} walk is \cs{1}.
For the \cs{k}ty of a non-trivial walk $W$, an added issue (see~\Cref{fig:examplecircular}) is that two different circular walks can cover the sea and cloud-related strongly connected components.
However, this is not possible if these components do not entirely cover the graph, i.e. the \river{} is not empty.
Thus, we get the following characterisation.
 
\begin{restatable}[$k$-Circular]{theorem}{kcircsafe}
\label{thm:kcircsafe}
A non-trivial walk $W$ is $k$-circular safe for $k \geq 2$, iff $\vapor(\heart(W))$ is a path and $\river(\heart(W))$ is non-empty.
\end{restatable}

Notice that our characterisation does not distinguish between $k$ when $k \geq 2$, implying that all the problems of $(k \geq 2)$-circular safety are equivalent.
These characterisations can be directly adapted for an optimal verification algorithm, by computing the hydrostructure for the given walk in linear time. This also results in $O(m^2n)$ time enumeration algorithms using the {\em two pointer} algorithm on a simple circular arc-covering walk of length $O(mn)$ (see \Cref{apx:closed:implementation}).
Moreover, using the optimal $O(mn)$ time \cs{1} algorithm~\cite{DBLP:journals/talg/CairoMART19}, the maximal \cs{k} walks can also be {\em optimally} enumerated in $O(mn)$ time, by computing the hydrostructure for $O(n)$ non-trivial \cs{1} walks (\Cref{thm:prelim}~\ref{item:omnitig-properties}) and using an interesting property of \cs{1} walks that are not \cs{k}.
See \Cref{apx:closed:implementation} for more details on the implementation.

\begin{figure}[t]
    \centering
    \centerline{
    \begin{subfigure}[m]{0.25\linewidth}
    \centering
    \includegraphics[trim=0 30 40 0, clip,width=\textwidth]{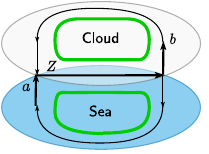}
    \caption{\label{fig:examplecircular}}
    \end{subfigure}
    \hspace{2em}
    \begin{subfigure}[m]{0.40\linewidth}
    \centering
    \includegraphics[trim=2 9 27 0, clip,width=\textwidth]{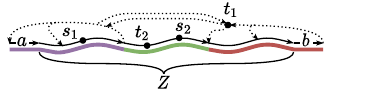}
    \caption{}
    \label{fig:linear-trivial-wings}
    \end{subfigure}
    }
    \caption{\small (a) In the absence of the $\river$, the two cycles (shown in thick green) can cover the cloud-related and sea-related SCCs of a non-trivial $\heart$ without having $aZb$ as subwalk.
    (b) The heart (green) and the wings (violet and red) of an avertible walk $W = aZb$:
    $W$ is unsafe for $s = s_1$ and $t = t_1$, because $\heart(W)$ is $s_1$-$t_1$ suffix-prefix covered in $W$ as it is part of the proper suffix of $W$ starting from $s_1$;
    $W$ is safe for $s = s_2$ and $t = t_2$, because the middle of $\heart(W)$ is not $s_2$-$t_2$ suffix-prefix covered in $W$, since it is contained neither in the suffix starting in $s_2$ nor in the prefix ending in $t_2$.
    }
    \label{fig:cirular}
\end{figure}

\begin{restatable}{lemma}{kcircularsubwalks}
\label{lem:multisafe-subwalks}
Let $k \geq 2$.
For a \cs{1} walk $W = X \heart(W) Y$ where $\heart(W)= aZb$, if $W$ is not \cs{k} then $XaZ$ and $ZbY$ are \cs{k} walks.
\end{restatable}
\section{Safety in Linear Models}
\label{s:open}

In this section we characterise the linear models for strongly connected graphs.
In the corresponding appendix we also extend them to non-strongly connected graphs.
Compared to the circular models, in the linear models the wings of a walk are not naturally safe, because $s$ or $t$ might occur in a wing. This allows a candidate solution to stop at $s$ or $t$, covering the graph without having the entire wing as a subwalk. However, only a \cs{1} walk $W$ can be \ls{k}. Otherwise, using an arc-covering closed walk that does not contain $W$, we can construct an arc-covering $s$-$t$ walk by repeating the arc-covering walk, once starting from $s$ and once ending at $t$.

We handle the wings of non-trivial \ls{k} walks separately.
Our characterisation has three conditions that force a walk to reach from $\sea$ to $\cloud$ of a bridge-like walk, and a fourth condition that makes only trivial walks safe.
For a trivial walk, an additional sufficient condition for safety is if there is an arc in its heart that can only be covered by using the whole walk, i.e. it cannot be covered by using only a proper suffix starting in $s$ or a proper prefix ending in $t$.
Formally we say that an arc $e$ in the heart of a trivial walk $W$ is \emph{$s$-$t$ suffix-prefix covered} if it is part of a proper prefix of $W$ that starts in $s$ or part of a proper suffix of $W$ that ends in $t$ (see \Cref{fig:linear-trivial-wings} for an example).
We get the following characterisation (see \Cref{apx:linear} for the proofs of this section).

\begin{theorem}[name = $k$-st Safe Trivial Walks and Non-Trivial Hearts, restate = linear]
    \label{thm:linear}
    A trivial walk $W$ (or a non-trivial heart $W$)
    is \ls{k} iff it is a single arc or
    \begin{enumerate}[label = (\alph*), nosep]
        \item $\river(W)$ cannot be covered with $k$ walks, or \label{thm:linear:river}
        \item $s \notin R^-(W)$, or \label{thm:linear:r-}
        \item $t \notin R^+(W)$, or \label{thm:linear:r+}
        \item $W$ is trivial and some arc $e$ in $\heart(W)$ is not $s$-$t$ suffix-prefix covered in $W$. \label{thm:linear:suffix-prefix}
    \end{enumerate}
\end{theorem}

Note that only \cs{1} walks can fulfill this characterisation: a trivial walk is always \cs{1}, and a non-trivial heart that is not \cs{1} fulfills none of the four criteria.

Further, note that Condition~\labelcref{thm:linear:river} also makes the wings safe, since they are not in the \river{} by \Cref{lem:clouds}~\labelcref{lem:clouds:connectedseacloud}.
Additionally if $s$ is in the left wing, we can cover the prefix of the left wing before $s$ without traversing the whole walk in two ways.
Firstly, if the prefix before $s$ is repeated in the suffix of the whole walk from $s$.
Note that this case only happens if $s$ is also in the heart, i.e. the left wing extends back into the heart, which makes the \sea{} of the heart a path and the \river{} empty.
Secondly, if $t$ prematurely ends the traversal before covering the whole walk, i.e. $t$ is present in the residual walk or the \sea{}.
We get the following characterisation.

\begin{theorem}[name = Wings of Non-Trivial $k$-st Safe Hearts, restate = linearntwings]
    \label{thm:linear-nt-wings}
    A non-trivial \cs{1} walk $W^lWW^r$, where $W$ is its heart, $W^l$ its left and $W^r$ its right wing, is \ls{k} iff $W$ is \ls{k} and
    \begin{enumerate}[label = (\alph*), nosep]
        \item $\river(W)$ cannot be covered with $k$ walks, or \label{thm:linear-nt-wings:antichain}
        \item either both of $s \in W^l$ and $t \in W^r$ are false, or exactly one is true with:
        \begin{itemize}[nosep]
            \item if $s \in W^l$, then $t \notin R^+(W) \cup W^r$ and $s \notin W$, and
            \item if $t \in W^r$, then $s \notin R^-(W) \cup W^l$ and $t \notin W$. 
        \end{itemize}
        \label{thm:linear-nt-wings:st}
    \end{enumerate}
\end{theorem}

Note that the \ls{k} problems are equivalent for all $k \geq m$.
Additionally, the maximal output size of the linear models is $O(n^3)$.
We can check a walk for \ls{k}ty in $O(m+f(m,n))$ where $f(m,n)$ is the time required to check the walk-cover.
And using the two-pointer algorithm on each \cs{1} walk, we can enumerate all maximal \ls{k} walks in $O(m^2n + mnf(m,n))$ time.
Recall from the preliminaries that $f(m,n) \in O(mn)$, and further observe that for $k \geq m$ or $k = 1$ we get $O(m)$.
See \Cref{apx:linear:implementation} for more details on the implementation.

\section{Incremental Computation of the Hydrostructure}
\label{sec:algorithms}

We have earlier described the optimal verification algorithms for various safety models using the hydrostructure. However, to report all maximal linear safe solutions optimally, we additionally need a more efficient computation of the hydrostructure for all the relevant walks. This is achieved by incremental computation of $R^+(\cdot)$, $R^-(\cdot)$ and the $\river$ as follows.

\paragraph{Incremental Maintenance of Forward and Backwards Reachability.}
\label{sub:incremental-hydrostructure}

Because of symmetry it suffices to describe how to compute $R^+(W)$.
Consider that for a bridge-like walk $aZb$, $R^+(aZb)$ is $aZ$ plus everything that is reachable from outgoing arcs of $aZ$ other than $b$ without entering $aZ$.
This allows us to describe $R^+(aZb)$ as the union of subsets {\em reachable} from the outgoing {\em siblings} of arcs on $Zb$, denoted as the {\em sibling reachability} \Rsib{e} of an arc $e \in Zb$, and formally defined as $\Rsib{e} := \{x \in G \mid \exists \text{ $\tail{}(e)$-$x$ walk in } G \setminus e\}$.
Formally we get the following (see \Cref{apx:implementation:incremental-hydrostructure-proofs} for the proofs of this paragraph).

\begin{lemma}[restate = incrementalcloudsrepresentation, name = ]
    \label{lem:simplified-clouds-computation}
    Given a bridge-like walk $aZb$ it holds that $R^+(aZb) = aZ \cup \bigcup_{\text{arc } e \in Zb} \Rsib{e}$.
\end{lemma}

Additionally, since subwalks of bridge-like walks are bridge-like, the following allows us to compute \Rsib{\cdot} incrementally along a bridge-like walk.

\begin{lemma}[restate = incrementalcloudscomputation, name = ]
    \label{lem:incremental-clouds}
    Let $aZb$ be a bridge-like walk where $b$ is a split-arc and $e$ a split-arc after $a$.
    Then
    \begin{enumerate}[label = (\alph*), nosep]
        \item $\Rsib{e} \subseteq \Rsib{b}$ and \label{lem:incremental-clouds:subset}
        \item $\Rsib{b} = \Rsib{e} \cup \{x \in G \mid \exists \text{ $e$-$x$ walk in } G \setminus (\Rsib{e} \cup \{b\})\}$. \label{lem:incremental-clouds:reachability}
    \end{enumerate}
\end{lemma}

Hence, we can compute all $R^+_{sib}(\cdot)$ using a single interrupted traversal. We number the nodes and arcs in $W$ and then annotate each node and arc $x$ with the first arc $e \in W$ for which it is in \Rsib{e}.
Afterwards we can decide if a node or arc $x$ is in $R^+(aZb)$ of a given subwalk $aZb$ by checking if $x$ enters \Rsib{\cdot} before or with $b$, or if $x \in aZ$.

\paragraph{Incremental Maintenance of the River Properties.}
\label{sub:incremental-river}

When using the two-pointer algorithm after computing the hydrostructure incrementally, we can maintain properties of the $\river$ used throughout the paper (\Cref{sec:subsetCov}).
By \Cref{lem:incremental-clouds} each node and arc enters $R^+(\cdot)$ (and leaves $R^-(\cdot)$) at most once during an execution of the two-pointer algorithm on a bridge-like walk.
We can therefore sort the nodes and arcs of $G$ into buckets in $O(m)$ time according to when they enter and leave the \river{}.
The existence of the \river{} can then be maintained in $O(m)$ time for a complete execution of the two-pointer algorithm.
Moreover, it is possible to maintain the in-degree and out-degree of each node and arc in the \river{} as well as the number of \emph{sources} (in-degree zero), \emph{sinks} (out-degree zero) and \emph{violators} (in-degree or out-degree at least two), and hence determine if the \river{} is a path. Further, \Cref{lem:incremental-clouds} also implies that entire SCCs of the \river{} (using both $\Rsib{\cdot}$ and $R^-_{sib}(\cdot)$) enter and leave the \river{} at most once during the execution of the two pointer algorithm. Hence, the SCCs of the \river{} can also be identified in overall $O(m)$ time. Finally, the size of a minimum walk-cover can also be maintained after each insertion or deletion of an edge from the \river{} using the {\em folklore} fully dynamic update algorithm for maximum flows~\cite{GuptaK18} in $O(m)$ time per update. Since each edge enters and leaves the \river{} at most once maintaining the size of the minimum walk-cover takes total $O(m^2)$ time.

\paragraph{Applications.}

The above incremental computation techniques result in improving the $O(m^2n)$ term of runtimes of linear problems by a factor of $m$ using amortisation as follows.
We execute the two-pointer algorithm on all maximal \cs{1} walks (which we get in $O(mn)$ time from~\cite{DBLP:journals/talg/CairoMART19}), but only maintain the hydrostructure and \river{} properties for non-trivial and bridge-like trivial \cs{1} walks (there are only $O(n)$ of these by \Cref{thm:max-omnitig-bounds} and they can be found in total $O(mn)$ time by \Cref{thm:bridge-like-check}~\labelcref{thm:bridge-like-check:check}).
For non-trivial \cs{1} walks, instead of maintaining the hydrostructure for the whole walk, we maintain it for the heart.
We describe how we handle trivial avertible walks in the same time bounds as non-trivial walks in \Cref{apx:linear:avertible-linear-implementation}.
Checking the walk-cover can be done as described above for $k = 1$ (maintain if the \river{} is a path) or $k \geq m$ (always false) in $O(m)$ time per maximal bridge-like walk.
Recall that the output size is potentially larger than $O(mn)$ resulting in the following.

\begin{theorem}
    \label{thm:improved-linear-runtimes}
    Given a strongly connected graph with $m$ arcs and $n$ nodes, all the maximal \ls{k} walks can be reported in $O(mn + n \cdot f(m,n) + |out|)$ time, where $f(m,n)$ is the time required for maintaining the size of the walk-cover ($f(m,n) = O(m)$ for $k = 1$ or $k \geq m$, and $O(m^2)$ otherwise), and $|out|$ is the size of the output.
\end{theorem}

\begin{figure}[thb]
    \centering
    \includegraphics[trim = 0 2 0 4, clip, scale=1]{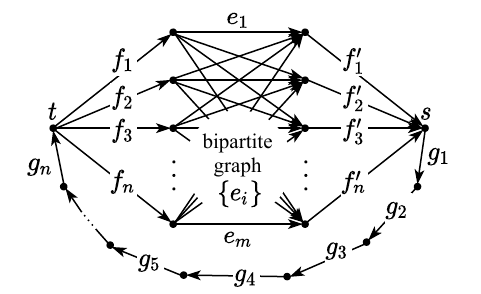}
    \caption{\small A graph class with $\Theta(n)$ nodes and $m \in \Omega(n) \cap O(n^2)$ arcs.
    The bipartite graph has $m$ arcs, resulting in total $\Theta(m)$ arcs.
    For each $i \in \{1, \dots, m\}$, $W_{i} = g_1 \dots g_n f_{e_i} e_{i} f'_{e_i} g_1 \dots g_n$ is \cs{k} and \ls{k} for each $k \geq 1$, and the total length of all $W_{i}$ is $\Omega(mn)$.
    This example is adapted from~\cite{DBLP:journals/talg/CairoMART19}.}
    \label{fig:runtime-tightness-example}
\end{figure}

The runtimes of $O(mn)$ are optimal because there exists a set of graphs (see \Cref{fig:runtime-tightness-example}) with safe walks in each model of total length $\Omega(mn)$.
As mentioned in \Cref{apx:linear:implementation}, the output size of the linear models is bounded by $O(n^3)$.

\section{Safety for Subset Covering and Subset Visibility}
\label{sec:subsetCov}

In this section we give a short overview of safety in $F$-covering and $F$-visible models where $F \subseteq E$.
For a full discussion and formal characterisation of these models, see \Cref{sec:subsetCovApx,sec:subsetVisApx}.

\paragraph{Subset covering.}
For the purpose of this section we denote an arc in $F$ by \emph{$F$-arc} and a subgraph or walk containing an $F$-arc with the predicate \emph{$F$-covered}.
Similar to how only subwalks of 1-circular safe walks can be safe in $E$-covering models, in subset covering only subwalks of walks that are safe in an $E$-covering model can be safe in the corresponding $F$-covering model.
In the $F$-covering models, single arcs in $E \setminus F$ are not trivially safe.
So in order to characterise their safety, we split them using a dummy node, since the hydrostructure is only defined on walks with at least two arcs.
Furthermore, if $F$ is empty, then no walk is safe, so in this section we assume $F \neq \emptyset$.

In the \emph{circular models}, a trivial walk is safe if its heart is $F$-covered.
Further, for bridge-like walks (including trivial walks), we define the \emph{$F$-covered SCCs} as the SCCs of the \river{} as well as the sea- and cloud-related SCCs (see \Cref{fig:Hydro-properties}) if they are $F$-covered.
Then a bridge-like walk is safe if any solution walk needs to leave an $F$-covered SCC, since then it needs to complete a full cycle through the hydrostructure, forcing it to have $W$ as subwalk.
This is required if there is an $F$-arc outside an SCC in the \river{} or the number of $F$-covered SCCs is greater than $k$.
In all other cases, the walk is not safe.

In the \emph{linear models}, we define the restricted reachability of a walk $W$ from the point of view of $s$ and $t$ denoted by $R^+_s(W)$ and $R^-_t(W)$.
We define $R^+_s(W)$ as the subgraph that is reachable from $s$ without using $W$ as subwalk and symmetrically $R^-_t(W)$ as the subgraph that reaches $t$ without using $W$ as subwalk.
Further, we refer to the set $R^+_s(W) \cap R^-_t(W)$ as the \emph{st-induced subgraph of $W$} and to $R^+_s(W) \cap R^-_t(W) \cap \river(W)$ as the \emph{st-induced \river{} of $W$}.
Lastly, we define the \emph{$F$~walk-cover} of a graph as a set of walks that cover all $F$-arcs in that subgraph.

With these definitions, we get a characterisation similar to that for the corresponding $E$-covering model.
We again characterise the non-trivial wings separately, but here we describe only the characterisation of trivial walks and non-trivial hearts.
For the latter two, the conditions on the walk cover of the \river{} and the $s$-$t$ suffix-prefix covering of an arc in the trivial heart are tightened:
we substitute the walk cover with an $F$ walk-cover, and the arc in the trivial heart with an $F$-arc.
The condition on the location of $s$ and $t$ is tightened by excluding the case where all $F$-arcs can be covered by $s$-$t$ walks without having $W$ as subwalk.

For the verification and enumeration algorithms, the sets $R^+_s(W)$ and $R^-_t(W)$ can be computed similarly to $R^+(W)$ and $R^-(W)$ in the same asymptotic time, and the $F$ walk-cover can be computed similarly to the walk cover as well, again in the same asymptotic time.
So all runtime bounds for $E$-covering models apply also to the $F$-covering models.

\paragraph{Subset visibility.}
In subset visibility, we limit the solution of a model to its visible arcs.
Thus we define $vis_F(W)$ as the subsequence of a walk $W$ which belongs to $F$.
Note that this does not change the solution of the problem, but only its representation, which is now limited to the visible set $F$.
So if a walk $W$ is a solution in an $E$-visible model, then $vis_F(W)$ is a solution in the corresponding $F$-visible model.
Abusing notation we also simply say that $W$ is a solution for the $F$-visible model of the problem.

To employ the hydrostructure for subset visibility, we relax the definition of the restricted reachabilities of a walk $W$.
We define $R^+_F(W)$ as everything reachable by a walk $W'$ from $\wh(W)$ such that $vis_F(W)$ is not a subwalk of $vis_F(W')$, and symmetrically $R^-_F(W)$ as everything that reaches $\wt(W)$ by a walk $W'$ such that $vis_F(W)$ is not a subwalk of $vis_F(W')$.
With this modified definition of the hydrostructure, the condition for the \vapor{} to be a path in \Cref{lem:cloud-cases} relaxes to the visible \vapor{} being a path.
Similarly, the other properties of the hydrostructure are adapted by considering $vis_F(W)$ instead of $W$.

In subset visibility, arcs are considered adjacent if they are connected by an invisible path.
This new \emph{visible} adjacency can be precomputed in $O(mn)$ time, such that all algorithms for $E$-visible models can be directly adapted to $F$-visible models with an added expense of $O(mn)$ time.

\section{Conclusion}
\label{sec:conclusion}

The issue of finding a safe and \emph{complete} genome assembly algorithm (for a natural problem formulation, namely that of a closed arc-covering walk) has been open for two decades, and recently solved in~\cite{tomescu2017safe}. On genome assembly graphs constructed from perfect data, omnitigs led to significantly longer strings than those of unitigs, proving the potential usefulness of a rigorous mathematical approach to genome assembly. However, because of many practical aspects of the read data (linear chromosomes, missing coverage, errors and diploidy) it remained open whether this theoretical breakthrough can also be made practical. A further issue with the existing results~\cite{tomescu2017safe,cairo2020macrotigs,acosta2018safe} is their very problem-specific techniques, which seem extremely difficult to generalise.

In this paper we presented an entirely {\em new perspective} to characterise safe solutions for such covering problems, with a unified methodology, having many implications. First, the hydrostructure of a walk highlights important reachability properties of a graph from the perspective of the safe walk. This results in simpler and {\em complete} characterisations (in the form of a YES-certificate) of the previously studied types of covering walks. Second, these characterisations are directly adaptable to optimal verification algorithms (except for \ls{k} with $k\neq 1,\infty$ in non-strongly connected graphs), and simple enumeration algorithms.
We also improved our \cs{k} and \ls{k} (for $k=1,\infty$) algorithms to optimality using the optimal \cs{1} algorithm~\cite{DBLP:journals/talg/CairoMART19} and the incremental computation of the hydrostructure, respectively. Third, the flexibility of the hydrostructure allowed us to easily handle further practical aspects. Namely, we easily extended the \emph{characterisations} for the subset covering model, for all \cs{k} and \ls{k} problems. Moveover, we easily generalised the \emph{hydrostructure} for the subset visibility model, so that the existing characterisations can be applied directly without changes. Lastly, even though the hydrostructure is developed mainly for arc-covering walks motivated by genome assembly, it reveals a novel fundamental insight into the structure of a graph, which can be of independent interest.

We deem that the extensive array of covering problems considered this paper, together with the simple characterisations and enumeration algorithms, proves that the hydrostructure is an ideal toolkit for obtaining safe, \emph{complete} and \emph{practical} genome assembly algorithms. As such, it could lead to another theoretical breakthrough with significant practical improvements, in the same way as the FM-index~\cite{FerraginaM00} is now at the core of start-of-the-art read mapping tools.

\subsubsection*{Acknowledgments}

We thank Manuel A. Cáceres for helpful comments on the manuscript's introduction.

\bibliographystyle{plain}
\bibliography{bibliography}

\newpage
\appendix

\section{Bioinformatics Motivation\label{sec:bioinformatics-motivation}}

Assuming we are sequencing a single circular genome (as when sequencing a single bacterium), the most basic notion of a solution (or a source genome) is that of a walk in the assembly graph covering all nodes or all arcs at least once, thereby explaining their existence in the assembly graph~\cite{tomescu2017safe}.
Even if the earlier works (mentioned in \Cref{sec:intro} on ``global'' genome assembly formulations) include various \emph{shortest} versions of such walks (e.g.~Eulerian, Chinese postman, or even Hamiltonian), it is widely acknowledged~\cite{nagarajan2009parametric,narzisi14,medvedev2019modeling,nagarajan2013sequence,DBLP:books/cu/MBCT2015,kingsford2010assembly} that a shortest walk misses repeats in the source genome. Moreover, safe walks for \emph{all} arc-covering walks are also trivially safe for more specialised types of walks. Both of these facts motivate the lack of constraints of the solution walks, except than to require that they are arc-covering.

However, such theoretical formulation of genome assembly using a single closed arc-covering walk uses the following assumptions. (i) All the reads are sequenced from a \emph{single circular} genome, such that these reads have (ii) \emph{no errors}, and (iii) \emph{no missing coverage} (e.g.~every position of the genome is covered by some read). However, these are very strong and impractical assumptions that are violated by real input data sets, as we explain next:

\begin{itemize}
\item Assumption (i) is not practical for several reasons. 
When sequencing a bacterium~\cite{charlebois1999organization}, one indeed obtains reads from a \emph{single circular} genome. However, when sequencing all bacteria in an environmental sample (as in \emph{metagenomics}~\cite{metagenomics}), the reads originate from \emph{multiple circular} genomes. In case of a virus~\cite{Gelderblom:1996aa}, the reads originate from a \emph{single linear} genome. Finally, when sequencing eukaryotes such as human~\cite{brown2002human}, the reads originate from \emph{multiple linear} chromosomes. 

\item Assumption (ii) is not practical because the sequencing process introduces errors in the reads, with error rates ranging from 1\% for short reads, up to 15\% for long reads. Such read errors produce certain known structures in the assembly graph~\cite{nagarajan2013sequence}, such as \emph{tips} (short induced paths ending in a sink), and \emph{bubbles} (two induced paths of the same length, and with the same endpoints, see also their generalisation to \emph{superbubbles}~\cite{IliopoulosKMV16,GartnerMS18,PatenERNGH18}), which are usually handled in an initial \emph{error-correction} stage using heuristics (e.g.~tips are removed, and bubbles are ``popped'' by removing one of the parallel paths). Such bubbles can also arise from a correct read position in diploid genomes (such as human), but on which the mother's copy of the chromosome differs from the father's. Popping bubbles is favorable in this setting in order to obtain longer assemblies (since unitigs would otherwise be broken up by the endpoints of the bubble). Moreover, in a de Bruijn graph, arcs that appear very few times (having low {\em abundance}) in the input reads are heuristically removed, as they are assumed to be errors.
But in practice, the assembly graph's topology might give evidence for the correctness of these arcs, especially in scenarios where low abundance is common~\cite{li2015megahit}, so a more accurate removal strategy is likely to result in a better assembly. This applies even to state-of-the-art long read assemblers like wtdbg2~\cite{Ruan:2020aa} which simply removes low abundance arcs.
\item Assumption (iii) is not true due to the practical limitations of current sequencing technologies. Thus, since not all parts of the genome can always be read, even if the sample contains a {\em single circular} genome, the reads appear as if they were sequenced from \emph{multiple linear} genomic sequences. 
\end{itemize}

\begin{figure}
    \centering
    \includegraphics{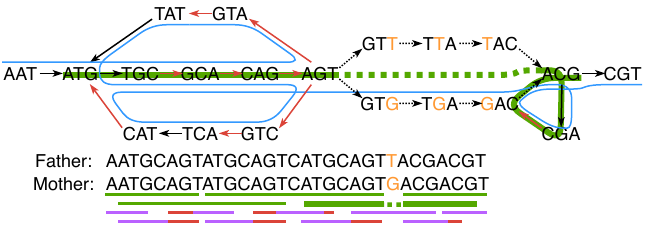}
    \caption{\small Assume a ``full'' read coverage scenario in which reads of length at least 4 have been sequenced starting from each position a diploid genome (i.e.~from both the mother and the father copy of a chromosome). An \emph{arc-centric} de Bruijn graph of order 3 has all substrings of length 3 of the reads as nodes, and all substrings of length 4 as arcs (from their length-3 prefix to their length-3 suffix). The thin blue line depicts the mother genome as a walk in this graph. Orange marks a position in which the two haplotypes differ; this induces a bubble in the de Bruijn graph (dotted arcs). We assume that the red arcs have low abundance in the read dataset for some natural reasons, despite being correct. But they would still be removed by current heuristic error-correction strategies, even though our models show evidence for their correctness.
    We set the dotted arcs to invisible, and model the genome as an $s$-$t$ walk not required to cover the red and dotted arcs, where $s = \textsf{AAT}$ and $t = \textsf{CGT}$. Below the mother's haplotype, we mark in green its maximal \ls{1} substrings. The thick safe walk of the mother's haplotype is also depicted in the graph (in green as well). We show in magenta the unitigs of the graph aligned to the mother haplotype. Additionally, the red parts of the unitigs would be missing if the red arcs were removed.
    }
    \label{fig:example-model-variants}
\end{figure}

The genome assembly models from \Cref{sec:intro} handle all such issues flexibly and theoretically solid by considering \emph{collections} of closed or open walks. In the subset covering models, the arcs not required to be covered can be those in tips, bubbles and those with low abundance in the reads. However, merely making them avoidable (and not removing them) can break the safety around such regions. Hence, another possible extension is to mark such parts of the graph as invisible. For example, marking bubbles as invisible prevents disrupting the safety of their flanking regions. 
See \Cref{fig:example-model-variants} for an example of these models.

Note that in practice k-mer multiplicities are available with the genome data, and these could be used to compute even longer safe paths when integrated into our models.
However this is beyond the scope of our paper, as they address a fundamentally different graph problem, and may result in the following major theoretical problems.

In case we use multiplicities to impose a coverage-goal for each arc, specifying exactly how often a valid solution walk must contain an arc, then the resulting theoretical formulation as a graph problem changes to a Eulerian cycle, rather than simply an edge covering cycle. Practically, due to errors in the data and varying abundance of different genomic regions, we may end up with ``non-Eulerian'' nodes, i.e. nodes that must be entered more or less often than they must be exited.
Using another formulation, the above issue can be addressed by allowing an interval of multiplicities for each arc, where the existence of a valid solution walk obeys these interval restrictions.
However, this can be reduced\footnote{Assign the closed interval $[0, 1]$ to each arc, and $[1,1]$ to each node $v$ extended as an arc $(v_{in},v_{out})$, where all incoming arcs of $v$ are incident on $v_{in}$ and the outgoing arcs of $v$ leave from $v_{out}$.} to the Hamiltonian circuit problem which is NP-Hard.
Our models address the information on k-mer abundances using subset covering and visibility, loosing out on some information to avoid modifying the underlying graph problem.

\section{Omitted Hydrostructure Proofs}
\label{apx:omitted-proofs}

\cloudcases*
\begin{proof}
    We decompose $W$ into $aZb$, where $a$ and $b$ are arcs.
    We first prove that if $\vapor(aZb)$ is not exactly the open path $Z$, then $\vapor(aZb) = G$ and $aZb$ is avertible.
    Assume that $Z$ is not an open path, or $\vapor(aZb)$ contains a node or arc $x \notin Z$.
    We distinguish three cases:
    \begin{itemize}[nosep]
        \item 
        If $Z$ is not an open path, then it contains a cycle which can be removed to create a walk from $a$ to $b$ not containing $aZb$.
    
        \item
        If $x=a$ (or $x=b$) then by definition $R^-(aZb)$ (or $R^+(aZb)$) contains a walk from $a$ to $b$ that does not contain $aZb$.

        \item
        If $x \notin \{a, b\}$, then by definition of $R^+(aZb)$ and $R^-(aZb)$, there are paths from $a$ to $x$ and from $x$ to $b$ that do not contain $aZb$.
        And since $x$ is not in $aZb$, these paths can be joined together to create a walk from $a$ to $b$ that does not contain $aZb$.
    \end{itemize}
    In each case, there is a walk from $a$ to $b$ that does not contain $aZb$.
    And since $b$ is the last element of this walk, by strong connectivity this walk can be extended within $G$ to reach any node or arc in $G \setminus Z$ without containing $aZb$.
    By definition, the walk $Z$ is in $\vapor(aZb)$ as well.
    So $R^+(aZb) = G$, and by symmetry one can prove that $R^-(aZb) = G$, so $\vapor(aZb) = G$.
    This also makes $aZb$ avertible, since, for any pair $x, y \in G$, every occurrence of $aZb$ in an $x$-$y$ walk can be replaced with an $a$-$b$ walk that does not contain $aZb$.
    
    We now prove that if $\vapor(aZb)$ is an open path equal to $Z$, then $aZb$ is {\em bridge-like}. 
    Since $Z$ is an open path, $b\notin Z$ else it would repeat the node $\tail(b)$. However, $b\in R^-(aZb)$ by definition, so we have $b\notin R^+(aZb)$, which means $aZb$ is bridge-like because of the pair $(a,b)$.
\end{proof}

\clouds*
\begin{proof}
    We prove the statements in reverse order for simplicity.
    \begin{itemize}
        \item[\labelcref{lem:clouds:connectivity}]
        We first prove all the incident pairs involving the  $\river(aZb)$, and then the residual pairs involving the $\vapor(aZb)$, and finally between $\sea(aZb)$ and $\cloud(aZb)$.
        
        \begin{itemize}
            \item
            
            We show that $x,y \in (\river, \vapor)$ and $x,y \in (\river, \cloud)$ do not exist.
            Since $x \in \river(aZb)$, it holds that $x \notin R^+(aZb)$ and therefore $x \neq a$.
            Thus, $y\in R^-(aZb)$ (and hence $\vapor(aZb)$ or $\cloud(aZb)$), implies $x\in R^-(aZb)$ because of $(x,y)$, as prepending $x$ ($\neq a$) does not make $aZb$ a subwalk.
            This contradicts that $x\in \river(aZb)$ by definition.
            By symmetry $x,y \in (\vapor, \river)$ and $x,y \in (\sea, \river)$ do not exist as well.

            \item
            For $x,y \in (\vapor, \cloud)$ we necessarily have $y = b$, because $x\in R^+(aZb)$ would otherwise add $y\in R^+(aZb)$, removing it from $\cloud(aZb)$. When $y= b$ we get $aZb$ as a subwalk, preventing $R^+(aZb)$ from covering $y$.
            By symmetry, $x,y \in (\sea, \vapor)$ implies $x = a$.

            \item
            We show that $x,y \in (\sea, \cloud)$ implies $x = a$ and $y=b$,
            which contradicts our definition of $aZb$ as being from arc to arc.
            This is because $x\in R^+(aZb)$ would add $y\in R^+(aZb)$ if $y\neq b$, as described above, and symmetrically $y\in R^-(aZb)$ implies $x= a$.
            
            Let $x,y\in (\cloud,\sea)$.
            We prove that if $x$ is an arc then $aZb$ is R-univocal (by symmetry if $y$ is an arc then $aZb$ is univocal), and hence trivial.
            Then $x \in R^-(aZb)$ but $y \notin R^-(aZb)$, so $x = b$ because each other arc is covered by $R^-(aZb)$ using its head.
            Assume for a contradiction that $aZb$ contains a join node.
            Then by strong connectivity an arc $e$ entering $aZb$ at this join node needs to be (forwards) reachable from $y=\head(b)$, let $W$ be such a walk of minimum length.
            But then, since $y \notin R^-(aZb)$, $W$
            (being a $b$-$e$ walk) needs to have $aZb$ as subwalk, which contradicts that it is minimal.
        \end{itemize}
        
        Where it applies, the necessity of the connections follows from {\em strong connectivity} of the graph since no other connections exist, as proved above.

        \item[\labelcref{lem:clouds:connectedseacloud}]
        Let $W^P$ be the prefix of $aZb$ ending in its last split node.
        We show that if $\sea(aZb) \neq \{a\}$, then for every node and arc $x \in \sea(aZb)$ with $x \neq a$ there is a path $P$ in $\sea(aZb)$ containing $x$ which starts from an arc $e$ with $\tail(e) \in \vapor(aZb)$ and ends in $a$. Notice that $\tail(e)$ is a split-node, since $P$ \emph{diverges} from $aZb$ at $e$, ensuring that $\tail(e)$ is in $W^P$.
        This will prove the first part of the statement because it allows to form a cycle through $a$ and $x \neq a$ within $W^P \cup \sea(aZb)$, because the path would leave $aZb$ latest at its last split node to enter $\sea(aZb)$ 
        (using \Cref{lem:clouds}~\ref{lem:clouds:connectivity}).
        And this will also prove the second part by symmetry.
        
        For a node or an arc $x \neq a$ in $\sea(aZb)$, by definition there is an $a$-$x$ walk inside $R^+(aZb)$ that does not have $aZb$ as subwalk.
        Starting from $\head(a) \in \vapor(aZb)$ this walk exits $\vapor(aZb)$ into the $\sea(aZb)$ (since $R^+(aZb)=\vapor(aZb) \cup \sea(aZb)$), with an arc $e$ having $\tail(e) \in \vapor(aZb)$. %
        By strong connectivity, there is also a path from $x$ to $a$. 
        Notice that both the paths from $e$ to $x$ and from $x$ to $a$ are completely within the $\sea(aZb)$ because the only arc leaving $\sea(aZb)$ is $a$, by \Cref{lem:clouds}~\ref{lem:clouds:connectivity}.
        
        This also proves that there exists such a circular path for $a$ if there exists some $x\in \sea(aZb)\neq a$.
        For $\sea(aZb)=a$, we have no split nodes on $aZb$ because $R^+(aZb)$ does not leave $aZb$. This makes $W^P=\emptyset$ resulting in no SCC.

        \item[\labelcref{lem:clouds:separation}]
        By definition, $a \in R^+(aZb)$ and $a \notin R^-(aZb)$ when $aZb$ is bridge-like, resulting in $a \in \sea(aZb)$. Symmetrically, we prove $b \in \cloud(aZb)$.

        \item[\labelcref{lem:clouds:backwards}]
        By strong connectivity, there always exists an $x$-$y$ path.
        If such a path contains $aZb$, then consider its prefix, say $P$, that ends at the first occurrence of $a$.
        Notice that if $\sea(aZb) = \{a\}$, then $y = a$, making $P$ a path from $x$ to $y$ without having $aZb$ as subwalk.
        Otherwise, by \Cref{lem:clouds}~\labelcref{lem:clouds:connectedseacloud}, there is a path from $a$ to every such $y \in \sea(aZb)$ which does not contain $b$, which when appended to $P$ gives us the desired walk with no subwalk $aZb$.

        \item[\labelcref{lem:clouds:forwards}]
        The only way to reach $\sea(aZb)$ from $\cloud(aZb)$ is through $\vapor(aZb)$ (by \Cref{lem:clouds}~\labelcref{lem:clouds:connectivity}) so every path from $\sea(aZb)$ to $\cloud(aZb)$ passes through $\vapor(aZb)$.
        Further, $\vapor(aZb)$ can only be entered from $\sea(aZb)$ through $a$ (by \Cref{lem:clouds}~\labelcref{lem:clouds:connectivity}) and can only be exited to $\cloud(aZb)$ through $b$ (by \Cref{lem:clouds}~\labelcref{lem:clouds:connectivity}), and every such path contains the entire $aZb$ as subwalk (by \Cref{lem:cloud-cases}).
        \qedhere
    \end{itemize}
\end{proof}
\section{Safety in Circular Models}
\label{apx:closed}

As described in \Cref{s:closed}, since trivial walks are naturally safe the characterisations are stated only for non-trivial walks.

\onecircsafe*
\begin{proof}
$(\Rightarrow)$
Let $\vapor(\heart(W))$ not be a path.
Let $aZb := \heart(W)$, $a'$ be a sibling join arc of $a$ and $b'$ be a sibling split arc of $b$.
Then, by \Cref{lem:cloud-cases}, $aZb$ is avertible.
We prove that $aZb$ is not \cs{1} by constructing a circular arc-covering walk $\mathcal{C}$ that does not contain $aZb$.
Let $P$ be an $a'$-$b'$ walk, such that $aZb$ is not a subwalk of $P$, which exists exactly because $aZb$ is avertible (if it does not exist, then $a'$ and $b'$ would make $aZb$ a bridge-like walk).
Let $\mathcal{C'}$ be an arbitrary circular arc-covering walk of $G$. If $\mathcal{C'}$ does not contain $aZb$ the claim holds, otherwise consider $\mathcal{C}$ obtained from $\mathcal{C'}$, by replacing every $aZb$ subwalk with $aZPZb$.
Concluding, since $aZb$ is not \cs{1}, it holds that $W$ is not \cs{1}.

$(\Leftarrow)$
Let $\vapor(\heart(W))$ be a path.
Then $\vapor(\heart(W))$ is not the whole $G$.
Moreover, by \Cref{lem:cloud-cases}, $\heart(W)$ is bridge-like, so there are $x, y \in G$ such that each $x$-$y$ walk has $\heart(W)$ as a subwalk.
Therefore, since every circular arc-covering walk $\mathcal{C}$ in $G$ contains both $x$ and $y$ and hence a subwalk from $x$ to $y$, $\heart(W)$ is a subwalk of $\mathcal{C}$ (and hence \cs{1}).
Further, since $\cal C$ is circular, the whole $W$ is \cs{1}, as $\heart(W)$ can only be traversed by also passing through its univocal extension.
\end{proof}

\kcircsafe*
\begin{proof}
$(\Rightarrow)$
We first prove that $W$ is not \cs{k} if 
$\vapor(\heart(W))$ is not a path or $\river(\heart(W))$ is empty.
\begin{itemize}
    \item If $\vapor(\heart(W))$ is not a path, then using \Cref{th:circularSafe} $W$ is not \cs{1}, or there is a circular arc-covering walk $\mathcal{C}$ not containing $W$.
    Thus, a collection of $k$ copies of $\mathcal{C}$ is also arc-covering not containing $W$, proving that $W$ is not \cs{k}.
    
    \item If $\vapor(\heart(W))$ is a path, and $\river(\heart(W))$ is empty, we can prove that $W$ is not \cs{k} for $k\geq 2$ because for a non-trivial $W$ it holds that $\cloud(\heart(W)) \cup \vapor(\heart(W))$ and $\sea(\heart(W)) \cup \vapor(\heart(W))$ are each strongly connected. Thus, two different walks can separately cover them as $\river(\heart(W))$ is empty.
    
    We only prove that $\cloud(\heart(W)) \cup \vapor(\heart(W))$ is strongly connected as the other follows by symmetry. Since $W$ is non-trivial, $\heart(W)$ starts with a join arc $a$ and ends with a split arc $b$. 
    If $\cloud(\heart(W))=\{b\}$ then clearly it is the join sibling of $a$, making $\cloud(\heart(W)) \cup \vapor(\heart(W))$ a cycle and hence strongly connected. Otherwise, if $\cloud(\heart(W))\neq\{b\}$ then the strong connectivity follows by \Cref{lem:clouds}~\labelcref{lem:clouds:connectedseacloud}.
\end{itemize}

$(\Leftarrow)$
If $\vapor(\heart(W))$ is a path then $\heart(W)$ is bridge-like by \Cref{lem:cloud-cases}.
If $\river(\heart(W))$ is not empty, since the graph is strongly connected we can enter $\river(\heart(W))$ only from $\cloud(\heart(W))$ (by \Cref{lem:clouds}~\labelcref{lem:clouds:connectivity}) %
using some arc. Hence, every walk covering this arc requires a subwalk from $\river(\heart(W))$ to $\cloud(\heart(W))$, which passes through $\sea(\heart(W))$ (by \Cref{lem:clouds}~\labelcref{lem:clouds:connectivity}). Using \Cref{lem:clouds}~\labelcref{lem:clouds:forwards} we get that this subwalk necessarily contains the walk $W$ making it \cs{k}.
\end{proof}

We also use the following property for our efficient algorithm to compute \cs{k} walks. 

\kcircularsubwalks*
\begin{proof}
By symmetry of the statement, we only prove that if $W$ is a 1-circular safe walk of $G$ that is not $k$-circular safe then $XaZ$ is a $k$-circular safe walk of $G$. This also implies that $Q$ is non-trivial otherwise it is implicitly \cs{k}.
Further, being a subwalk of the \cs{1} $W$ makes $Q := XaZ$ also \cs{1}.
Let $\heart(Q) = aZ^*b^*$, which is either $aZ$ or a prefix of $aZ$. 
Since $aZb$ is a non-trivial heart, there exists a sibling split arc $b'$ of $b$.
We claim that for any such sibling $b'$, we have $b' \in \river(aZ^*b^*)$, which is equivalent to $b' \notin R^+(aZ^*b^*)\cup R^-(aZ^*b^*)$. 

Now, $aZ$ is bridge-like (since $aZb$ is bridge-like), so each $\wt{aZ}$-$\wh{aZ}$ walk has $aZ$ as subwalk.
Since $\wh{aZ} = \tail(b')$ it follows that each $a$-$b'$ walk has $aZ$ as subwalk, and hence its prefix $aZ^*b^*$. Thus, $b' \notin R^+(aZ^*b^*)$. Further, we know that $b'\in \sea(aZb)$ as it is surely in $R^+(aZb)$ (by definition) and not in $\vapor(aZb)$ (for bridge-like $aZb$ by \Cref{lem:cloud-cases}). So any path from $b'$ to $b^*\in \vapor(aZb)$ surely enters $\vapor(aZb)$ through $a$ (by \Cref{lem:clouds}~\ref{lem:clouds:connectivity}). Hence it has $aZ^*b^*$ as a subwalk since $\vapor(aZb)$ is a path and $aZ^*b^*$ is a prefix of $aZb$. Thus, $b'\notin R^-(aZ^*b^*)$.

Thus $b' \notin (R^+(aZ^*b^*) \cup R^+(aZ^*b^*))$ implying that $b'\in \river(aZ^*b^*)$, and by \Cref{thm:kcircsafe} that $aZ^*b^*$ is \cs{k}.
Therefore, since $\heart(XaZ) = aZ^*b^*$, we conclude that $XaZ$ is \cs{k}. By symmetry $ZbY$ is also \cs{k}.
\end{proof}

\subsection{Implementation}
\label{apx:closed:implementation}
The characterisation can be directly adapted to get a verification algorithm for evaluating whether a given walk is safe in $1$-circular and $k$-circular models. A trivial walk $W$ is easily identified in $O(|W|)$ time by checking for univocal walks. For a non-trivial walk $W$ the hydrostructure can be computed in linear time (\Cref{thm:hydro_algo}) and the characterisation can be verified in $O(|W|)$ time resulting in an $O(m)$ time verification algorithm for both \cs{1} and \cs{k} walks.

Since every trivial walk is \cs{1} and \cs{k}, they can be enumerated by computing univocal extensions of each arc in $O(mn)$ time. For reporting all non-trivial maximal \cs{1} and \cs{k} walks, we use the two pointer algorithm on a candidate solution.  A candidate solution can be computed by arranging the $m$ arcs in some circular order, and adding the shortest path between adjacent arcs to complete the closed arc-covering walk of size $O(mn)$. To get an arc-covering solution made up of $k$ closed walks, we can have $k$ \emph{identical} copies of this $O(mn)$-size walk. Now, using the two-pointer algorithm (\Cref{thm:prelim}~\labelcref{thm:two-pointer-algirithm}) on one of these identical walks, with the above linear time verification algorithm we can enumerate all non-trivial maximal \cs{1} and \cs{k} walks in $O(m^2n)$ time. 

Alternatively, non-trivial \cs{k} walks can computed using non-trivial \cs{1} walks in $O(mn)$ time. This is because the amount of non-trivial \cs{1} walks is $O(n)$ (\Cref{thm:prelim}~\ref{item:omnitig-properties}). Computing their hydrostructure in $O(m)$ time either verifies them to be \cs{k} or its two subwalks as \cs{k} (using \Cref{lem:multisafe-subwalks}). Hence, using the $O(mn)$ time optimal \cs{1}  algorithm~\cite{DBLP:journals/talg/CairoMART19} we get the following result.
\begin{theorem}
    Given a strongly connected graph with $m$ arcs and $n$ nodes, all the maximal \cs{k} walks can be reported in $O(mn)$ time.
\end{theorem}
\section{Safety in Linear Models}
\label{apx:linear}

Our characterisation employs the notion of {\em $s$-$t$ suffix-prefix covered} for the safety of trivial walks. It describes whether an arc of a walk is in its suffix starting at $s$ or in its prefix ending at $t$. This notion thus identifies arcs that cannot be covered without covering the whole walk, rendering it safe. We formally define it using {\em proper $s$-suffix} and {\em $t$-prefix} as follows.

\begin{definition}[Proper $s$/$t$-Suffix/Prefix]
    For a walk $aZb$, the \emph{proper $s$-suffix (or $t$-prefix)} is its maximal suffix starting in $s$ (or prefix ending in $t$) or the empty sequence if $s \notin Z$ (or $t \notin Z$). Recall that $Z$ starts and ends with a node.
\end{definition}

\begin{definition}[$s$-$t$ Suffix-Prefix Covering]
    An arc $e\in aZb$ is \emph{$s$-$t$ suffix-prefix covered} in $aZb$ if $e$ belongs to the $s$-suffix or to the $t$-prefix of $aZb$.
\end{definition}

Furthermore, we use the following property of the trivial avertible walks in our main proof.

\begin{lemma}
    \label{lem:vapor-only-clouds-covering}
    For a trivial avertible walk $aZb$ in a strongly connected graph $G$, the graph $G \setminus \heart(aZb)$ is also strongly connected.
\end{lemma}
\begin{proof}
    Observe that since $\heart(aZb)$ is from arc to arc and $\heart(aZb)$ is a biunivocal path, removing $\heart(aZb)$ from $G$ does not remove any endpoint of remaining arcs.
    For any two nodes $x,y$ in $G \setminus \heart(aZb)$, since $aZb$ is avertible there is an $x$-$y$ (or $y$-$x$) walk that does not have $aZb$ as a subwalk. Thus, the graph $G\setminus \heart(aZb)$ is strongly connected.
\end{proof}

We now prove our characterisation of \ls{k} walks, by first giving a characterisation that works for trivial walks and non-trivial hearts. For simplicity of presentation, we will separately characterise the safety of the wings of non-trivial walks. 

\linear*
\begin{proof}
    If $W$ is a single arc, it is trivially safe.
    Otherwise, let $W=aZb$ where $a$ and $b$ are arcs. 
    
    \noindent $(\Rightarrow)$
    When the conditions~\labelcref{thm:linear:river,thm:linear:r-,thm:linear:r+,thm:linear:suffix-prefix} are false, we show that $aZb$ is not safe as follows.
    Note that if \labelcref{thm:linear:suffix-prefix} is false, then at least one of $s$ and $t$ are in $aZb$.
    We distinguish the cases where the walk is avertible or bridge-like.
   
    When $aZb$ is avertible, if it is trivial, then by \Cref{lem:vapor-only-clouds-covering} there exists an arc-covering circular walk $\mathcal{W}$ in $G \setminus \heart(aZb)$ because it is a strongly connected component.
    This circular walk can be extended to get a linear $s$-$t$ walk starting from a suffix of $aZb$ and (or) ending at a prefix of $aZb$.
    Since all arcs in $\heart(aZb)$ are $s$-$t$ suffix-prefix covered in $aZb$, this walk covers all the arcs of $\heart(aZb)$ without having $aZb$ as a subwalk.
    Thus, $aZb$ is not \ls{k}.
    On the other hand, if $aZb$ is an non-trivial heart, then $aZb$ is not \cs{1} by \Cref{th:circularSafe}. Hence, there exists a circular arc-covering walk in $G$ that does not have $aZb$ as subwalk.
    Such a walk can be repeated such that it starts in $s$ and ends in $t$, forming an arc-covering $s$-$t$ walk that does not have $aZb$ as subwalk.
    Thus, $aZb$ is not \ls{k}.
    
    For bridge-like $aZb$ we describe a set of $k$ $s$-$t$ walks that together cover $G$, where none has $aZb$ as subwalk.
    Since $s \in R^-(aZb)$ and $t \in R^+(aZb)$, an $s$-$t$ walk can cover the cloud-related SCC and the sea-related SCC (\Cref{lem:clouds}~\labelcref{lem:clouds:connectedseacloud}), necessarily covering $\cloud(aZb)$ and $\sea(aZb)$.
    For a non-trivial heart it then also covers $\vapor(aZb)$.
    For trivial walks, all arcs in $\heart(aZb)$ can be covered by the proper $s$-suffix and (or) $t$-prefix of $aZb$, since they are $s$-$t$ suffix-prefix covered in $aZb$.
    Finally, $\river(aZb)$ can be completely covered by the $k$ walks between $\cloud(aZb)$ and $\sea(aZb)$, because \labelcref{thm:linear:river} is false.
    Hence, $G$ can be covered without going from $\sea(aZb)$ to $\cloud(aZb)$ and, therefore, without having $aZb$ as a subwalk.
    Thus, $aZb$ is not \ls{k}.

    \noindent $(\Leftarrow)$
    If $\river(aZb)$ cannot be covered with $k$ walks, then to cover $\river(aZb)$ at least one walk needs to enter $\river(aZb)$ twice.
    This ensures that it has a subwalk from $\sea(aZb)$ to $\cloud(aZb)$ and hence $aZb$ (by \Cref{lem:clouds}~\labelcref{lem:clouds:forwards}).
    If $s \notin R^-(aZb)$, then the only way to cover $b$ (and $\cloud(aZb)$) is through a walk from $\sea(aZb)$ to $\cloud(aZb)$, necessarily having $aZb$ as subwalk (by \Cref{lem:clouds}~\labelcref{lem:clouds:forwards}). This can be symmetrically proven for $t\notin R^+(aZb)$.
    Finally, if $aZb$ is a trivial walk with an arc in $\heart(aZb)$ that is not $s$-$t$ suffix-prefix covered in $aZb$, then it can only be covered by using $aZb$ as a whole.
\end{proof}

\begin{lemma}
    \label{lem:wings-in-sea-cloud}
    A non-trivial \cs{1} walk $W^lWW^r$ with $W = \heart(W^lWW^r)$ has $W^l$ in $R^+(W)$ and $W^r$ in $R^-(W)$.
\end{lemma}
\begin{proof}
    Since $W^lWW^r$ is non-trivial, $W$ is non-trivial and therefore $\sea(W) \neq \{\wt(W)\}$ from which by \Cref{lem:clouds}~\labelcref{lem:clouds:connectedseacloud} follows that $\sea(W)$ with some part of $\vapor(W)$ forms an SCC.
    But then, every R-univocal walk ending in $\wt(W)$ is in $\sea(W) \cup \vapor(W) = R^+(W)$ and $W^l$ is a subwalk of at least one of these walks.
    By symmetry, $W^r$ is in $R^-(W)$.
\end{proof}

\linearntwings*
\begin{proof}
    If the $\river(W)$ has cannot be covered with $k$ walks, then to cover it using only $k$ solution walks at least one walk enters $\river(W)$ twice.
    By \Cref{lem:clouds}~\labelcref{lem:clouds:forwards,lem:clouds:connectivity} this walk makes $W^lWW^r$ safe as its subwalk necessarily goes from $\sea(W)$ to $\cloud(W)$ which completely contain the wings (by \Cref{lem:wings-in-sea-cloud}). 
    
    Now, the wings are also safe if $s\notin W^l$ and $t\notin W^r$, as they are univocal extensions.
    They can also be safe with $s\in W^l$, if we are required to separately cover the prefix of the left wing before $s$,
    and the location of $t$ does not allow the walk to end prematurely before covering the entire $W^lWW^r$.
    Note that if $s\in W^l\cap W$ such a prefix would be covered starting from $s$ itself, negating the above requirement making the wings not safe.
    Further, covering the entire $W^lWW^r$ requires such a walk to have a subwalk from $\sea(W)$ to $\cloud(W)$ without ending before completing the right wing. This can be ensured by $t\notin R^+(W)\cup W^r$, else it can also end in $\vapor(W)$, $\sea(W)$ (as $W$ is non-trivial) or before completely covering $W^r$.
    By symmetry we have the corresponding criteria for $t\in W^r$.
\end{proof}

\subsection{Non-Strongly Connected Graphs}

In contrast to the circular models, linear models do not require the graph to be strongly connected to admit a solution.
Our hydrostructure (and hence the characterisations described earlier) handle only strongly connected graphs, but we can extend them to a graph $G$ which is not strongly connected but admits a solution.
Consider the SCCs of $G$ and observe that single nodes can also be SCCs.
We refer to an arc as {\em inter-SCC} if its endpoints belong to different SCCs of $G$, otherwise we refer to it as \emph{intra-SCC arc}.
Trivially, in the 1-linear model, the graphs have the following structure as a single $s$-$t$ walk must cover the whole graph:

\begin{lemma}
    \label{lem:1-linear-graph-structure}
    Let $G = V \cup E$ be a graph, and let $s,t \in V$. It holds that $G$ admits an $s$-$t$ arc-covering walk if and only if the DAG of SCCs of $G$ is a path, such that between any two SCCs consecutive in this path there is exactly one inter-SCC arc, $s$ is in the (only) source SCC and $t$ in the (only) sink SCC.
\end{lemma}

By processing the SCCs and the inter-SCC arcs of non-strongly connected graphs separately, we can reduce the characterisations and algorithms for the 1-linear model in non-strongly connected graphs to the 1-linear model in strongly connected graphs.
For that we define the \emph{entry points} of an SCC as $s$ or nodes that have an incoming intra-SCC arc, and the \emph{exit points} of an SCC as $t$ or nodes that have an outgoing intra-SCC arc.
Note that each SCC of a graph that admits a solution to the 1-linear model has exactly one entry point and exactly one exit point, by \Cref{lem:1-linear-graph-structure}.
This allows us to decompose the graph into a sequence $D := (S_1, e_1, S_2, \dots, e_{\ell - 1}, S_\ell)$ where each $S_i$ is an SCC and each arc $e_i$ is an inter-SCC arc.
We call this the \emph{sequence decomposition} of the graph.
Observe that each $S_i$ is of one of the following three forms: it can be a single node, a cycle of at least one arc, or a nonempty graph that is not a cycle.
Single nodes are a special case, as they allow to join their predecessor and successor arcs into a single safe walk.
Cycles are a special case as well, as our characterisations on strongly connected graphs are defined only for non-cycles, so we need to handle these separately as well.%
Finally, for a nonempty SCC $S_i$ that is not a cycle, it holds that the \ls{1} walks for $G$, that are completely contained in $S_i$, are also safe for single $s'$-$t'$ arc-covering walks in $S_i$, where $s'$ and $t'$ are the entry and exit points of $S_i$, respectively.
The sequence structure allows us to decompose solution walks into their subwalks through each of the elements in the sequence.
Note that the sequence structure also implies an order in which the arcs in the graph need to be covered.
This allows us to generalise the univocal extension for inter-SCC arcs to completely describe safe walks that contain inter-SCC arcs.
Additionally, these also have all safe walks in circular SCCs as subwalks (as proven in \Cref{thm:notsc-1-linear} below), solving this special case.

We define the \emph{inter-SCC univocal extension} $U'(e)$ of an inter-SCC arc $e$ as a walk $XYZ$ where $Y$ is from arc to arc and $X$ and $Z$ are within single SCCs and three further conditions hold.
First, $Y$ contains $e$ and contains only inter-SCC arcs and single-node SCCs.
Second, $X$ is the longest walk ending in $\wt(Y)$ that is R-univocal in the SCC of $\tail(\wt(Y))$ (ignoring inter-SCC arcs) and that visits the entry point of its SCC at most twice, and if it visits it twice then the first occurrence is $\wt(X)$.
And third, $Z$ is the longest walk starting in $\wh(Y)$ that is univocal in the SCC of $\head(\wh(Y))$ (ignoring inter-SCC arcs) and that visits the exit point of its SCC at most twice, and if it visits it twice then the second occurrence is $\wh(Z)$.
We get the following characterisation:

\begin{theorem}[1-st Safe Walks in General Graphs]
    \label{thm:notsc-1-linear}
    In a non-strongly connected graph $G$ that admits a solution to the single linear model for given $s$ and $t$, a walk is \ls{1} iff:
    \begin{enumerate}[label = (\alph*), nosep]
        \item it is a subwalk of the inter-SCC univocal extension of an inter-SCC arc, or \label{thm:notsc-1-linear:inter}
        \item it is 1-s't' safe in an SSC $S$ of $G$, where $s'$ is the entry point of $S$, and $t'$ is the exit point of $S$. \label{thm:notsc-1-linear:intra}
    \end{enumerate}
\end{theorem}
\begin{proof}
    We start by handling case \ref{thm:notsc-1-linear:intra}.
    If $W$ is \ls{1} in an SCC of $G$ that is not a cycle where $s$ is the entry point of $S$ and $t$ is the exit point of $S$, then $W$ is part of each solution walk in $G$ and therefore safe.
    On the other hand, if $W$ is not \ls{1} in its SCC, then it can be avoided by a solution walk in $G$, and is therefore unsafe.
    This proves the theorem for walks that are in an SCC that is not a cycle, so we exclude those walks for the rest of the proof.
    
    To prove that $W$ is safe if it is a subwalk of the inter-SCC univocal extension $U'(e) = XYZ$ of an inter-SCC arc $e$, we prove that $XYZ$ is safe.
    By definition $Y$ is biunivocal and therefore safe.
    We prove the claim that after traversing $Y$ each solution walk needs to traverse $Z$, from which the safety of $XYZ$ follows by symmetry.
    By definition, $Z$ is in a single SCC $S$ and $\wt(Z)$ is its unique entry point.
    If $Z$ does not contain the exit point of $S$, then the claim holds.
    On the other hand, if $Z$ does contain the exit point of $S$, then it can ignore it the first time: by the strong connectivity of $S$ the exit point has an outgoing arc within $S$, and this arc cannot be covered before reaching the exit point for the first time, but it must be covered before exiting $S$.
    If $Z$ contains the exit point a second time, then by definition $Z$ ends there, so the exit point cannot end $Z$ early.
    Thus, after traversing $XYZ$ each solution walk needs to traverse~$Z$, which proves the safety of $YZ$, and by symmetry the safety of the entire $XYZ$.
    
    Now it remains to prove that there are no other safe walks.
    We first consider walks that contain an inter-SCC arc, and then all remaining walks.
    
    If $W$ contains an inter-SCC arc $e$, but is not a subwalk of its inter-SCC univocal extension $O(e) = XYZ$, then without loss of generality we can assume that a proper prefix of $W$ has $\wh(Y)Z$ as subwalk (the case where a proper suffix of $W$ has $X\wt(Y)$ as subwalk is symmetric).
    Let $S$ be the SCC of $Z$.
    Then $W$ is unsafe if $Z$ ends in a node $v$ that has at least two outgoing arcs in $S$: since $S$ is strongly connected, when a solution walk first reaches $v$ it can continue with either of the outgoing arcs in $S$ and then cover the rest of $S$.
    On the other hand, if $Z$ ends in the exit point of $S$ that is not a split node in $S$, then, by definition, $Z$ contains it twice.
    But then, all nodes in $S$ have exactly one outgoing intra-SCC arc, and therefore it is a cycle.
    So between the two occurrences of the exit point in $Z$, $S$ is completely covered, proving that there is a solution walk that exits $S$ after traversing $\wh(Y)Z$, in addition to a solution walk that traverses $S$ one more time after traversing $\wh(Y)Z$.
    Therefore, $W$ is not safe.
    
    If $W$ contains no inter-SCC arc, is not a subwalk of any inter-SCC univocal extension and is not fully contained in an SCC different from a cycle, then it is contained in a single SCC $S$ that is a cycle.
    Let $U'(e) = XYZ$ be the inter-SCC univocal extension of the inter-SCC arc $e$ that enters $S$.
    By definition, $Z$ is inside $S$, and since $W$ is not a subwalk of this inter-SCC univocal extension, it is not a subwalk of $Z$.
    But $Z$ is the maximal safe subwalk in $S$, since there is a solution walk that traverses $\wh(Y)Z$ and then exits $S$, as proven above.
    Therefore, $W$ is not safe.%
\end{proof}

In the $\infty$-linear model, there exists a simple reduction from non-strongly connected graphs to strongly connected graphs.

\begin{theorem}[$\infty$-st Safe Walks in General Graphs]
    \label{thm:notsc-infty-linear}
    A walk is \ls{\infty} in a non-strongly connected graph $G$ that admits a solution to the infinite linear model for given nodes $s$ and $t$ iff the walk is \ls{\infty} in $G \cup \{(t, s)\}$, which is strongly connected.
\end{theorem}
\begin{proof}
    Since $G$ can be covered by a set of $s$-$t$ walks, it holds that $G \cup \{(t, s)\}$ is strongly connected.
    For each set of solution walks in $G$, a set of solution walks in $G \cup \{(t, s)\}$ can be constructed by repeating one of the walks through $(t, s)$.
    Similarly, for each set of solution walks in $G \cup \{(t, s)\}$, a set of solution walks in $G$ can be constructed by splitting all walks that contain $(t, s)$.
    Therefore, a walk in $G$ is \ls{\infty} in $G$ if and only if it is \ls{\infty} in $G \cup \{(t, s)\}$.
\end{proof}

Using \Cref{thm:notsc-1-linear}, the \ls{1} algorithm for strongly connected graphs (see \Cref{thm:improved-linear-runtimes}) can be extended trivially to get an algorithm with the same time bounds for non-strongly connected graphs.
The algorithm solves all the \ls{1} subproblems in SCCs and adds the inter-SCC univocal extensions to the solution.
Since each SCC intersects with at most two inter-SCC univocal extensions of which one starts at its entry point and one ends at its exit point, the safe walks within each SCC can be checked for being subwalks of the inter-SCC univocal extensions by trivial alignment.
The total length of the inter-SCC univocal extensions that intersect with that SCC (limited to the part within the SCC) is $O(n_i)$ since the SCC is not a cycle and the intersecting parts are univocal or R-univocal.
By the same argument, the intersecting parts contain each arc at most once.
Therefore, the unique starting point of an alignment can be found in constant time after $O(n_i)$ time preprocessing (by building an array that maps from arc ids in the SCC to their positions in the intersecting parts of the inter-SCC univocal extensions).
Afterwards the alignment is checked in time linear in the length of the safe walk in the SCC, resulting in a total runtime of $O(m_in_i + o_i)$ where the SCC has $n_i$ nodes, $m_i$ arcs and the total length of its safe walks is $o_i$.
The overhead for finding the SCCs and extending the inter-SCC arcs is $O(m+n)$ time, and solving the subproblems including filtering subwalks takes an additional $O((m_1n_1 + o_1) +(m_2n_2 + o_2) + \dots)$ time which is subset of $O(mn + o)$ (where $o$ is the size of the output) since the $m_i$, $n_i$ and $o_i$ add up to $m$, $n$ and $o$ respectively.
For \ls{\infty}ty by \Cref{thm:notsc-infty-linear} we can use the $O(mn + o)$ algorithm for strongly connected graphs in addition to linear preprocessing and postprocessing.
Thus, we get the following theorem.

\begin{theorem}
    In a graph with $n$ nodes and $m$ arcs that admits a solution to the single (or infinite) linear model, all maximal \ls{1} (or \ls{\infty}) walks can be enumerated in $O(mn + o)$ time where $o$ is the size of the output.
\end{theorem}

\begin{remark}
For $k$ $s$-$t$ arc-covering walks where $1<k<\infty$, our characterisation and algorithms do not simply extend to non-strongly connected graphs. Note that a reduction similar to \ls{1} does not suffice as the graph can be any DAG of SCCs, which is not necessarily a path. 

Hence the SCCs may have up to $k$ {\em different} entry and $k$ {\em different} exit points. This is essentially an entirely different model for linear-safe walks from multiple sources and multiple sinks, which would require further case analysis and generalisations to 
our current characterisations of linear models. 
This becomes particularly complex because of the possible existence of such sources or sinks in different components of the hydrostructure, or wings of the walk. Even though we deem this to be feasible, it may require an entirely new approach and thus it is beyond the scope of this work.
\end{remark}

\subsection{Implementation}
\label{apx:linear:implementation}

Note that the \ls{k} problems are equivalent for all $k \geq m$.
The characterisation can be directly adapted to get a verification algorithm for evaluating whether a given walk is \ls{k}. A trivial walk $W$ is easily identified in $O(|W|)$ time by checking for univocal walks. The hydrostructure can be computed in linear time (\Cref{thm:hydro_algo}) and thereafter the $s$ and $t$ membership queries are answered in constant time. The $s$-$t$ suffix-prefix criteria can be verified in $O(|W|)$ time. Thus, the verification algorithm for \ls{k}ty requires $O(m+f(m,n))$ time, where time required to compute the minimum walk-cover is $f(m,n)$.

For enumerating all the maximal \ls{k} walks we require the set of all \cs{1} walks, which are computable in $O(mn)$ time~\cite{DBLP:journals/talg/CairoMART19}.
We compute the hydrostructure for each non-trivial $\heart(W)$ as well as each trivial $W$, and check the conditions from \Cref{thm:linear,thm:linear-nt-wings}.
If a walk is \ls{k}, we report it.
Otherwise, we use the two-pointer algorithm to report its maximal safe substrings as in \Cref{thm:prelim}~\labelcref{thm:two-pointer-algirithm}.

For each $W$, this algorithm performs safety checks linear in the length of $W$.
So since the total length of all maximal \cs{1} walks is $O(mn)$ by \Cref{thm:prelim}~\labelcref{item:omnitig-properties}, at most $O(mn)$ walks need to be checked for safety.
A safety check includes building the hydrostructure in $O(m)$ time, and checking the minimum walk-cover of the \river{}.
The latter check can be performed in $O(m)$ for $k = 1$ since the check simplifies to checking if the \river{} is a path, and for $k \geq m$ the check is always false.
Otherwise, we use the $O(mn)$ flow algorithm mentioned in \Cref{sec:preliminaries}.

Additionally, the output size of this algorithm is limited.
A trivial avertible walk $W$ has maximal \ls{k} subwalks of total length $O(|W|)$.
All other walks have length at most $O(n)$ each by \Cref{thm:prelim}~\labelcref{item:omnitig-properties}, and therefore all maximal \ls{k} subwalks of these have a total length of $O(n^2)$.
Since the total length of trivial avertible walks is $O(mn)$, and there are $O(n)$ other walks (by \Cref{thm:max-omnitig-bounds}), this results in an $O(mn + n^3) = O(n^3) \leq O(m^2n)$ bound for the total length of maximal \ls{k} walks.
Resulting, we get the following theorem.

\begin{theorem}
    The problems \ls{k} can be solved in $O(m^2n + mnf(m,n))$ time, or in $O(m^2n)$ for $k = 1$ or $k \geq m$.
\end{theorem}

\subsection{Handling Trivial Avertible Walks in the Implementation}
\label{apx:linear:avertible-linear-implementation}

We prove that we can enumerate the maximal safe subwalks of maximal trivial avertible \cs{1} walks in the linear models in $O(mn)$ time plus $O(n)$ executions of the two-pointer algorithm using the incremental hydrostructure, which matches the time bounds for computing \ls{k} from non-trivial \cs{1} walks for the respective model.
We first restrict our description to $E$-covering linear models.
Consider that for a trivial avertible walk $W^lWW^r$ where $W^l$ is its left wing, $W^r$ is its right wing and $W$ is its heart by \Cref{thm:bridge-like-check}~\labelcref{thm:bridge-like-check:subwalks} only subwalks of $W^lW$ or $WW^r$ can be bridge-like.
Therefore, the \ls{k}ty of other subwalks of $W^lWW^r$ can be decided without the hydrostructure by checking the location of $s$ and $t$ within the subwalk.
This can be done for all such subwalks in amortised $O(|W^lWW^r|)$ time using the two-pointer algorithm, so in total $O(mn)$ time for all walks.

For those subwalks that are subwalks of $W^lW$ or $WW^r$, we describe only $W^lW$ since the argument for $WW^r$ is symmetric.
Consider that $W^lW$ is trivial, so by the conditions on the heart in the linear models (see \Cref{thm:linear}~\labelcref{thm:linear:suffix-prefix}), it may have non-\ls{k} subwalks only if $W^lW$ contains $s$.
It does not matter if it contains $t$, since any subwalk of $W^lW$ has its last arc in its heart, so this last arc cannot be in a $t$-prefix.
So we only need to build the hydrostructure on $W^lW$ if it contains $s$, as it is safe otherwise.
These $W^lW$ can be enumerated in total $O(mn)$ time.

Assuming no multiarcs, all the $W^lW$ that contain $s$ form a tree because they are $R$-univocal. Since they all are different, with distinct leaf arcs they also have distinct head nodes. Thus, there exist total $O(n)$ such walks, for which we use the two-pointer algorithm with the incremental hydrostructure to find its maximal safe subwalks. Note that if $G$ has multiarcs, $W^l$ is clearly multiarc-free, and each parallel arc would form a distinct $W$ (containing a single arc), where each such walk is equivalent in terms of safety. Thus, we use the two-pointer algorithm for all of them together, executing it only $1$ time.
As a result we get the following theorem.

\begin{theorem}
    \label{thm:linear-avertible-algo}
    The maximal \ls{k} walks of a strongly connected graph $G$ that are subwalks of trivial avertible \cs{1} walks can be enumerated in $O(mn)$ time plus $O(n)$ executions of the two-pointer algorithm using the incremental hydrostructure variant for the specific model.
\end{theorem}

For $F$-covering, we need to make two fixes.
Again, we only describe the algorithm for $W^lW$, since for $WW^r$ it is symmetric.
In the case where $W^lW$ does not contain $s$, in contrast to the $E$-covering models above, this does not guarantee its safety in all cases (see \Cref{thm:linear}~\labelcref{thm:linear:suffix-prefix} or \Cref{thm:coveringlinear}~\labelcref{thm:coveringlinear:suffix-prefix}).
From those walks we exclude those that contain $t$ using $O(n)$ two-pointer executions by the same argument as for $s$ (as described above).
Then we also exclude those that do not contain an $F$-arc in their heart, since these are safe if they do not contain $s$ or $t$.
Of the remaining walks $W^lW$, only subwalks of $W^l$ can be safe, as the whole graph can be covered without using $W$ by \Cref{lem:vapor-only-clouds-covering}.
And since $W^l$ is maximal there can be at most $n$ of such subwalks $W^l$ uniquely defined by their last node.

For those walks $W^lW$ that contain $s$ and where $W$ has parallel arcs (as described above), the safety is now only equivalent for two parallel arcs if they both are in $F$ or both are not in $F$.
Hence, in such a case, we execute the two-pointer algorithm twice, once for arcs in $F$ and once for arcs not in $F$.
We get the following theorem.

\begin{theorem}
    \label{thm:coveringlinear-avertible-algo}
    The maximal \lcs{k} walks of a strongly connected graph $G$ that are subwalks of trivial avertible \cs{1} walks can be enumerated in $O(mn)$ time plus $O(n)$ executions of the two-pointer algorithm using the incremental hydrostructure variant for the specific model.
\end{theorem}

\section{Omitted Hydrostructure Implementation Details}
\label{apx:implementation}

\subsection{Basic Construction Algorithm}
\label{apx:implementation:hydrostructure}

We prove the correctness of \Cref{thm:hydro_algo} using \Cref{alg:hydrostructure}.

\begin{algorithm}[tbh]
	\caption{Forward Reachability}
	\label{alg:hydrostructure}
	\KwIn{Graph $G$, non-empty walk $W$}
	\KwOut{$R^+(W)$}
	\DontPrintSemicolon
	\BlankLine
	\If{$W$ is an open path}{
	$R' \gets \{x \in V \cup E \mid \exists \text{ $\wt{}(W)$-$x$ walk in } G \setminus \wh{}(W)\}$ \tcp*{compute e.g.~via BFS}
	\If{$ \not\exists \text{ arc } e \in R': e \notin W \land \head{}(e) \in W$\label{alg:hydrostructure:if}}{
	    \Return $R'$\; \label{alg:hydrostructure:returnr}
	}}
    \Return $G$ \; \label{alg:hydrostructure:returng}
\end{algorithm}

\hydrostructurebasicalgo*
\begin{proof}
    Each node and arc can be annotated with its membership in $R^+(W)$ and $R^-(W)$ using \Cref{alg:hydrostructure} and its symmetric variant in $O(m)$ time.
    It remains to prove the correctness of \Cref{alg:hydrostructure}.
    
    \begin{itemize}
        \item
        Suppose \Cref{alg:hydrostructure} reaches \Cref{alg:hydrostructure:returng} and returns $G$. 
        If $W$ is not an open path, then by \Cref{lem:cloud-cases} this is correct.
        Otherwise, if $W$ is an open path, there is an arc $e \in R'$ such that $\tail(e) \notin W \land \head(e) \in W \land e \neq \wt(W)$.
        Then, there is a walk from $\wt(W)$ via $e$ to $\wh(W)$ that does not contain $W$ as subwalk, so $W$ is avertible, and by \Cref{lem:cloud-cases} it follows that outputting $G$ is correct, because $G = \vapor(W) \subseteq R^+(W) \subseteq G$.
        
        \item
        Suppose \Cref{alg:hydrostructure} reaches \Cref{alg:hydrostructure:returnr} and returns $R'$. 
        Then for that to be correct, by \Cref{lem:cloud-cases} $W$ must be bridge-like and $R'$ must be equal to $R^+(W)$.
        
        Assume for a contradiction that $W$ is avertible.
        Then by \Cref{lem:cloud-cases} it holds that $\vapor(W) = G$, and therefore $R^+(W) = G$.
        Hence, by definition of $R^+(\cdot)$, there is a $\wt(W)$-$\wh(W)$ walk $\mathcal{W}$ that does not have $W$ as subwalk.
        Such a walk leaves $W$ before reaching $\wh(W)$ to avoid having $W$ as subwalk.
        But then it must enter $W$ again at an arc $e$ other than $\wt(W)$ before it reaches $\wh(W)$ for the first time.
        Therefore, $e \in R'$ and it holds that $e \notin W$ and $\head(e) \in W$.
        But then the condition in \Cref{alg:hydrostructure:if} would be false and \Cref{alg:hydrostructure} would not reach \Cref{alg:hydrostructure:returnr}, a contradiction.
        
        Therefore, $W$ is bridge-like, and it remains to prove that $R' = R^+(W)$.
        $(\subseteq)$
        Let $x \in R'$.
        Then there is a $\wt(W)$-$x$ walk in $G \setminus \wh(W)$, so $x \in R^+(W)$.
        $(\supseteq)$
        Let $x \in R^+(W)$.
        Then there is a $\wt(W)$-$x$ walk in $G$ that does not contain $W$ as subwalk.
        This walk cannot contain $\wh(W)$, since otherwise one of its prefixes until $\wh(W)$ would prove $\wh(W) \in R^+(W)$, contradicting \Cref{lem:clouds}~\labelcref{lem:clouds:separation}.
        So $x \in R'$.
        \qedhere
    \end{itemize}
\end{proof}

\subsection{Omitted Incremental Hydrostructure Proofs}
\label{apx:implementation:incremental-hydrostructure-proofs}

\incrementalcloudsrepresentation*
\begin{proof}
    $(\subseteq)$
    Let $x \in R^+(aZb)$.
    If $x \in aZ$, then it trivially holds, so consider $x \notin aZ$.
    Then, since $aZb$ is bridge-like, $x \in \sea(aZb)$.
    So by \Cref{lem:clouds}~\labelcref{lem:clouds:connectedseacloud}, there is an $a$-$x$ path $P$ in $\sea(aZb) \cup \vapor(aZb)$.
    Let $v$ be the last element of $P$ in $\vapor(aZb)$; it exists since $\head(a) \in \vapor(aZb)$, and since $\vapor(aZb)$ is a path by \Cref{lem:cloud-cases}, $v$ is a node.%
    Let $e'$ be the outgoing arc of $v$ in $P$ and $e$ its sibling in $Zb$; notice that $e' \neq b$, since $b \in \cloud(aZb)$.
    By definition of \Rsib{\cdot}, it holds that the suffix of $P$ starting from $e'$ is in \Rsib{e} and, as such, $x \in \Rsib{e}$.

    $(\supseteq)$
    By definition, $aZ \subseteq R^+(aZb)$.
    Let $x \in \Rsib{e}$, for some arc $e$ from $Zb$.
    Since $aZb$ is bridge-like, by \Cref{lem:cloud-cases}, $Z$ is an open path, and therefore $aZ$ is a path.
    As a result, there is a path $P$ from $a$ to $\tail{}(e)$ via a prefix of $Z$ that does not contain $e$.
    Since $x \in \Rsib{e}$, $P$ can be extended to $x$ without using $e$.
    So there is a walk from $a$ to $x$ without using $e \in Zb$, so $x \in R^+(aZb)$.
\end{proof}

\incrementalcloudscomputation*
\begin{proof}\leavevmode
    \begin{itemize}
        \item[\labelcref{lem:incremental-clouds:subset}]
        Let $x \in \Rsib{e}$.
        Then by \Cref{lem:simplified-clouds-computation}, $x \in R^+(aZb)$, so there is a walk from $a$ to $x$ which, by \Cref{lem:clouds}~\labelcref{lem:clouds:separation}, does not include $b$.
        Moreover, by \Cref{lem:clouds}~\labelcref{lem:clouds:connectedseacloud}, there is a walk from each split-sibling of $b$ to $a$ that does not include $b$.
        So $x \in \Rsib{b}$.
        
        \item[\labelcref{lem:incremental-clouds:reachability}]
        $(\subseteq)$
        Assume for a contradiction, that $x \in \Rsib{b}$ but not in $\Rsib{e} \cup \{x \in G \mid \exists \text{ $e$-$x$ walk in } G \setminus (\Rsib{e} \cup \{b\})\}$.
        Since $x \in \Rsib{b}$, there is a $\tail{b}$-$x$ walk $W$ in $G \setminus \{b\}$.
        Moreover, since by \Cref{lem:cloud-cases} it holds that $Z$ is an open path, it follows that $aZ$ does not contain $b$.
        Therefore, we can prepend a suffix of $aZ$ to $W$ (without repeating the last node of $Z$), resulting in an $e$-$x$ walk $W'$ in $G \setminus \{b\}$.
        Since $x \notin \Rsib{e}$ and $W'$ starts in $e$, $W'$ must contain a suffix starting in some occurrence of $e$ that contains no elements of $\Rsib{e}$.
        But then that suffix of $W'$ is an $e$-$x$ walk in $G \setminus (\Rsib{e} \cup \{b\})$, a contradiction.
        
        $(\supseteq)$
        If $x \in \Rsib{e}$ then $x \in \Rsib{b}$ by \Cref{lem:incremental-clouds}~\labelcref{lem:incremental-clouds:subset}.
        Otherwise there is an $e$-$x$ walk in $G \setminus (\Rsib{e} \cup \{b\})\}$.
        Since $aZb$ is bridge-like, it holds that each split sibling of $b$ is in $\sea(aZb)$.
        By \Cref{lem:clouds}~\labelcref{lem:clouds:connectedseacloud} there is a walk from each split sibling of $b$ to $e$ in $\sea(aZb) \cup Z$.
        Notice that by \Cref{lem:clouds}~\labelcref{lem:clouds:separation}, $b \in \cloud(aZb)$, so these walks do not contain $b$ since $Z$ is $\vapor(aZb)$.
        So there exists an $\tail{}(b)$-$x$ walk via $e$ in $G \setminus \{b\}$, so $x \in \Rsib{b}$.
        \qedhere
    \end{itemize}
\end{proof}

\subsection{Supporting Theorems for the Incremental Algorithms}

\begin{theorem}
    \label{thm:max-omnitig-bounds}
    A strongly connected graph $G$ with $n$ nodes and $m$ arcs has at most $m$ maximal \cs{1} walks, of which at most $O(n)$ are bridge-like trivial and at most $O(n)$ are non-trivial.
\end{theorem}
\begin{proof}
    Let $W$ be a trivial bridge-like walk.
    Then there is a pair of elements $x, y \in G$ such that each $x$-$y$ walk has $W$ as subwalk.
    Since $W$ is trivial this is equivalent to each $x$-$y$ walk containing some edge $e \in \heart(W)$.
    Therefore the removal of $e$ would disconnect $y$ from $x$, which, since $G$ is strongly connected, proves that $e$ is a strong bridge.
    Moreover, by~\cite{italiano2012finding}, there are at most $O(n)$ strong bridges, so there are at most $O(n)$ maximal bridge-like trivial walks.
    The bound on non-trivial \cs{1} walks follows from \Cref{thm:prelim}~\labelcref{item:omnitig-properties}.
\end{proof}

\begin{lemma}[restate = italianopath, name = ]
    \label{lem:italianopath}
    Let $W$ be a trivial $a$-$b$ walk in a strongly connected graph with $n$ nodes and $m$ arcs and let $e \in \heart(W)$ be an arc with $e \notin \{a, b\}$.
    After $O(m)$ time preprocessing, the query \emph{``is there an $a$-$b$ path in $G \setminus e$''} can be answered in $O(1)$-time.
\end{lemma}
\begin{proof}
    The query in \Cref{thm:prelim}~\ref{thm:prelim:georgiadis} (for $x=\head{a}$ and $y=\tail{b}$)
    is equivalent to answering whether there exists an $a$-$b$ path in $G\setminus e$, because we can prove that there is always a $b$-$a$ path that does not contain $e$ if $W$ is trivial.
    A $b$-$a$ path exists since $G$ is strongly connected.
    And this path does not contain $e$, as $W$ is trivial so every path through $e\in \heart(W)$ contains the whole $W$ starting from $a$.
\end{proof}

\begin{theorem}[restate = bridgelikecheck, name = ]
    \label{thm:bridge-like-check}
    Let $W=W^lZW^r$ be a trivial walk where $\heart(W)=Z$.
    \begin{enumerate}[label = (\alph*), nosep]
        \item \label{thm:bridge-like-check:check} Using $O(m)$ preprocessing time on $G$, we can verify if $W$ is avertible in $O(|W|)$ time.
        \item \label{thm:bridge-like-check:subwalks} If $W$ is avertible, then its only maximal bridge-like subwalks are $W^lZ$ and $ZW^r$.
    \end{enumerate}
\end{theorem}
\begin{proof}
    \begin{itemize}
        \item[\labelcref{thm:bridge-like-check:check}]
        Using \Cref{lem:italianopath}, we can check if an arc in $Z$ is part of every $\wt(W)$-$\wh(W)$ path in $O(1)$ time after $O(m)$ preprocessing.
        Since $Z$ is a path, the property of a single arc in $Z$ being part of every $\wt(W)$-$\wh(W)$ path is equivalent to the property of $Z$ being subwalk of every $\wt(W)$-$\wh(W)$ walk.
        And since $W$ is trivial, that is equivalent to each $\wt(W)$-$\wh(W)$ walk containing $W$ as subwalk.
        Resulting, since we can find the heart of a trivial walk in $O(|W|)$ time, the claim holds.
        
        \item[\labelcref{thm:bridge-like-check:subwalks}]
        Let $W' = W'^jZW'^s$ be a subwalk of $W$, with $W'^j, W'^s$ not empty.
        Then $\wt(W'^j)$ is a split node.
        Let $e$ be one of its outgoing arcs $e$ that is not in $Z$.
        By strong connectivity, there is an $e$-$\wt(W)$ path, which would repeat $\wt(W)$ if it would have $Z$ as subwalk, and therefore does not have $Z$ as subwalk.
        Prepending this path with $W'^j$ creates a $\wt(W')$-$\wt(W)$ path that does not have $Z$ as subwalk.
        Symmetrically, there is a $\wh(W)$-$\wh(W')$ path without $Z$ as subwalk.
        Additionally, since $W$ is avertible, there is a $\wt(W)$-$\wh(W)$ walk that does not have $Z$ as subwalk.
        Since neither of these three walks starts or ends with a node or arc in $Z$, they can be joined into a single $\wt(W')$-$\wh(W')$ walk that does not have $Z$ as subwalk, which by \Cref{lem:cloud-cases} proves that $W'$ is avertible.
        \qedhere
    \end{itemize}
\end{proof}

\section{Safety for Subset Covering}
\label{sec:subsetCovApx}

The subset covering models reduce the covering constraint to only a subset $F\subseteq E$ instead of the entire $E$. Note that this implies that any walk which is a solution in the $E$-covering model is also a solution in the $F$-covering model, but we can have walks which are not $E$-covering but $F$-covering. Thus, every $F$-covering safe walk is $E$-covering safe. Furthermore, if $F$ is empty then no walk is safe, so in this section we assume $F \neq \emptyset$.

Note that, in contrast to the $E$-covering models, in the $F$-covering models, a single arc is not trivially safe.
So in order to characterise their safety we extend the definition of hydrostructure for single arcs, by splitting the arc using a dummy node. Thus, a single arc can be treated as a walk with two arcs, for which the hydrostructure is already defined.

We call an arc in $F$ an $F$-arc and a part of the graph \emph{$F$-covered} if an $F$-arc belongs to it.
Also, for a bridge-like walk, we define the \emph{$F$-covered SCCs}, which are the maximal SCCs of its \river{}, and the sea- and cloud-related SCCs (see \Cref{lem:clouds}~\ref{lem:clouds:connectedseacloud}) that are $F$-covered.
Similarly, for an avertible walk, the \emph{$F$-covered SCC} is $G$.
Furthermore, the $F$-covered sea and cloud-related SCCs are counted as one if they share all their $F$-arcs.

\paragraph{Circular Models.}

A trivial walk $W$ is circular safe if $\heart(W)$ is $F$-covered. 
For bridge-like walks, the issue with $F$-covering circular walks is that a set of cycles can cover the entire $F$ if it is covered by some $F$-covered SCCs, making the walks unsafe. This is not possible if an arc in $F$ exists outside the $F$-covered SCCs, or the number of such SCCs are more than $k$.

\begin{theorem}[name = Circular $F$-Covering Safety, restate = coveringcircular]
    For $F\subseteq E$ a \cs{1} walk $W$ that contains at least two arcs is \ccs{k} iff for the hydrostructure on $W$ (for trivial) or $\heart(W)$ (for non-trivial) it holds that
    \begin{enumerate}[label = (\alph*), nosep]
        \item $k$ is less than the number of $F$-covered SCCs, or
        \item there exists an $F$-covered arc in the $\river$, that is not in any SCC of the $\river$, or
        \item $W$ has an $F$-covered trivial heart.
    \end{enumerate}
\end{theorem}
\begin{proof}
    $(\Rightarrow)$
    Assume that $k$ is greater than or equal to the number of $F$-covered SCCs, there exists no $F$-covered arc in the \river{}, that is not in any SCC of the \river{}, and $W$ does not have an $F$-covered trivial heart.
    We distinguish between $W$ being trivial avertible, trivial bridge-like or non-trivial.
    \begin{itemize}
        \item
        If $W$ is trivial and avertible but does not have an $F$-covered trivial heart, then by \Cref{lem:vapor-only-clouds-covering} all arcs in $F$ can be covered by a single circular walk without traversing $\heart(W)$.
        Thus, $W$ is not \ccs{k}.
        
        \item
        If $W$ is trivial and bridge-like or non-trivial (in which case $\heart(W)$ is bridge-like by \Cref{th:circularSafe}), then all arcs in $F$ are in $F$-covered SCCs. Indeed arcs in the \river{} are in $F$-covered SCCs by hypothesis; arcs in the \cloud{} or \sea{} are relatively in the $F$-covered cloud-related SCC or sea-related SCC; arcs in the \vapor{} of non-trivial walks are both in the sea-related SCC and cloud-related SCC; and for trivial walks the arcs in the \vapor{} that are neither in the cloud-related SCC nor in the sea-related SCC are in the trivial heart by definition, which does not contain arcs in $F$ by hypothesis.
        Therefore, since all arcs in $F$ are in $F$-covered SCCs and we have at most $k$ $F$-covered SCCs of which none contains the whole walk $W$, it follows that $F$ can be covered with $k$ walks not having $W$ as subwalk, hence $W$ is not \ccs{k}.
    \end{itemize}

    $(\Leftarrow)$
    \begin{itemize}
        \item
        If $W$ has an $F$-covered trivial heart then it is \ccs{k}.
        
        \item
        If $k$ is less than the number of $F$-covered SCCs, then at least one solution walk needs to intersect two SCCs.
        We need to distinguish three cases: (A) if these are sea-related and cloud-related (in which case the sea-related SCC and the cloud-related SCC do not share all their $F$-arcs), (B) if one is in the \river{} and the other is sea-related or cloud-related, or (C) if both are in the \river{}.
        In all cases, \Cref{lem:clouds}~\labelcref{lem:clouds:connectivity} implies that this solution walk needs to traverse from \sea{} to \cloud{}, which by \Cref{lem:clouds}~\labelcref{lem:clouds:forwards} implies the \ccs{k}ty of $W$.
        In case (C) this is because the \river{} is not strongly connected if there is more than one maximal SCC.
        
        \item
        If there exists an $F$-covered arc in the \river{}, that is not in any SCC of the \river{}, then the \river{} is not strongly connected.
        Therefore, to cover this arc with a circular walk, the \river{} needs to be exited and entered, which by \Cref{lem:clouds}~\labelcref{lem:clouds:connectivity} implies that a solution walk covering this arc needs to traverse from \sea{} to \cloud{}, which by \Cref{lem:clouds}~\labelcref{lem:clouds:forwards} implies the \ccs{k}ty of $W$.
        \qedhere
    \end{itemize}
\end{proof}

\paragraph{Linear Models.}

The characterisation for safe walks in the linear models can be translated to the $F$-covering model using the following definitions.
We define the $R^+_s(W)$ as the subgraph reachable from $s$ without using $W$ as a subwalk.
Similarly, we define $R^-_t(W)$ as the subgraph that reaches $t$ without using $W$ as a subwalk.
We refer to the set $R^+_s(W) \cap R^-_t(W)$ as the \emph{st-induced subgraph of $W$} and similarly $R^+_s(W) \cap R^-_t(W) \cap \river(W)$ as the \emph{st-induced \river{} of $W$}.
Note that if $s\in R^-(W)$ (or $t\in R^+(W)$) then $R_s^+(W)$ (or $R_t^-(W)$) covers the whole graph.
Also, if both $s,t\in \river(W)$ then the st-induced subgraph of $W$ contains $\vapor(W)$ disconnected from some part, if any, of $\river(W)$.
Further, we define the \emph{$F$ walk-cover} of a subgraph as a set of walks that cover all $F$-arcs in that subgraph, where the number of walks is referred as the \emph{size} of the $F$ walk-cover. Similar to a walk-cover (see \Cref{sec:preliminaries}), the $F$ walk-cover can be computed in $O(mn)$ time\footnote{In the minimum flow, we keep {\em unit} demand only for the edges in $F$.} using flows. We have the following  simple characterisation of safe walks in this model.

\begin{theorem}[name = $k$-st $F$-covering Safe Trivial Walks and Non-Trivial Hearts, restate = coveringlinear]
    \label{thm:coveringlinear}
    A \cs{1} walk $W$ that contains at least two arcs and that is trivial, or is a non-trivial heart, is \lcs{k} iff
    \begin{enumerate}[label = (\alph*), nosep]
        \item The $F$-arcs in the st-induced \river{} of $W$ cannot be covered with $k$ walks, or \label{thm:coveringlinear:river}
        \item $F \cup \{s, t\} \not\subseteq$ a single WCC in the st-induced subgraph of $W$, or \label{thm:coveringlinear:r}
        \item $W$ is trivial and some arc in $F \cap \heart(W)$ is not $s$-$t$ suffix-prefix covered in $W$. \label{thm:coveringlinear:suffix-prefix}
    \end{enumerate}
\end{theorem}
\begin{proof}
    Assume that all three conditions are false.
    \begin{itemize}
        \item 
        If $W$ is a non-trivial heart, then given \ref{thm:coveringlinear:r} is false, there exist $s$-$t$ walks not having $W$ as its subwalk. Further, for each $f \in F$ which shares the WCC with $s$ and $t$, there is an $s$-$t$ walk containing $f$ (as $f$ is in the st-induced subgraph of $W$) without having $W$ as subwalk. 
        And by \ref{thm:coveringlinear:river} using up to $k$ $s$-$t$ walks all $F$-arcs can be covered without any such walk having $W$ as subwalk.

        \item
        If $W$ is a bridge-like trivial walk, then similar to the non-trivial $\heart$ case all $F$-arcs outside $\heart(W)$ can be covered with $k$ $s$-$t$ walks that do not have $W$ as subwalk.
        Additionally, since by \ref{thm:coveringlinear:suffix-prefix} all arcs in $F \cap \heart(W)$ are $s$-$t$ suffix-prefix covered in $W$, these $k$ walks also cover the $F$-arcs inside $\heart(W)$.
        
        \item
        If $W$ is an avertible trivial walk, then all $F$-arcs in its heart are $s$-$t$ suffix-prefix covered in $W$, so by \Cref{lem:vapor-only-clouds-covering} there exists an $s$-$t$ walk that covers all $F$-arcs without having $W$ as subwalk.
    \end{itemize}
    
    Assume that one condition is true at a time.
    \begin{itemize}
        \item
        If the $F$-arcs in the st-induced \river{} of $W$ cannot be covered with $k$ walks then at least one walk needs to exit and reenter the river, making $W$ safe.
        
        \item
        If $F \cup \{s, t\}$ is not completely contained in a single WCC of the st-induced subgraph of $W$, then there is no $s$-$t$ walk without $W$ as subwalk (if $\{s, t\}$ are not in the st-induced subgraph of $W$), or for an $F$-arc $f$ there is no $s$-$t$ walk that contains $f$ without $W$ as subwalk (if $f$ is not in the st-induced subgraph of $W$ or not in same WCC as $\{s,t\}$ in the st-induced subgraph of $W$).
        
        \item 
        If $W$ is trivial and some arc in $F \cap \heart(W)$ is not $s$-$t$ suffix-prefix covered in $W$, then to cover that arc an $s$-$t$ walk needs to enter $W$ at its start and leave it at its end, making $W$ safe.
        \qedhere
    \end{itemize}
    
\end{proof}

\begin{theorem}[name = Wings of Non-Trivial $k$-st $F$-covering Safe Hearts, restate = coveringlinearntwings]
    \label{thm:coveringlinear-nt-wings}
    A non-trivial walk $W^lWW^r$, where $W$ is its heart, $W^l$ its left and $W^r$ its right wing and where $W$ contains at least two arcs is \lcs{k} iff $W$ is \lcs{k} and
    \begin{enumerate}[label = (\alph*), nosep]
        \item The $F$-arcs in the st-induced \river{} of $W$ cannot be covered with $k$ walks, or
        \label{thm:coveringlinear-nt-wings:river}
        
        \item Either both of $s \in W^l$ and $t \in W^r$ are false, or exactly one is true with: \label{thm:coveringlinear-nt-wings:st}
        \begin{itemize}[nosep]
            \item if $s \in W^l$, then $t \notin R^+(W) \cup W^r$ and $F\not\subseteq R^-_t(W) \cup $ suffix of $W_l$ starting in $s$, and
            \item if $t \in W^r$, then $s \notin R^-(W) \cup W^l$ and $F\not\subseteq R^+_s(W) \cup $ prefix of $W_r$ ending in $t$.
        \end{itemize}
    \end{enumerate}
\end{theorem}
\begin{proof}
Similar to the non-trivial hearts, not being able to cover the $F$-arcs of the st-induced \river{} of $W$ with $k$ walks ensures that one walk needs to exit and re-enter the river, making $W$ safe along with its wings as such a walk goes from $\sea(W)$ to $\cloud(W)$ which completely contain the wings (by \Cref{lem:wings-in-sea-cloud}). 

Also, $s\notin W^l$ and $t\notin W^r$ makes the wings safe if $W$ is safe as they are R-univocal and univocal, respectively. %
They can also be safe with $s\in W^l$, if we are required to separately cover the prefix of the left wing before $s$,
and the location of $t$ and the $F$-arcs does not allow the walk to end prematurely before covering the entire $W^lWW^r$. Note that if $F$ is entirely contained in $R_t^-(W)$ and suffix of $W_l$ starting in $s$, $F$ would be covered starting $W_l$ only from $s$ itself, negating the above requirement making the wings not safe.
    Further, covering the entire $W^lWW^r$ requires such a walk to have a subwalk from $\sea(W)$ to $\cloud(W)$ without ending before completing the right wing. This can be ensured by $t\notin R^+(W)\cup W^r$, else it can also end in $\vapor(W)$, $\sea(W)$ (as $W$ is non-trivial) or before completely covering $W^r$.
    By symmetry we have the corresponding criteria for $t\in W^r$.
\end{proof}

\paragraph{Implementation.}
The computation of $F$-covered SCCs, $R^+_s(W)$ and $R^-_t(W)$ clearly requires linear time in the size of the graph, i.e., $O(m)$. The $F$ walk-cover can be computed using flows in $O(mn)$ time as described earlier. Furthermore, the number of executions of the two-pointer algorithm are still $O(n)$ (see \Cref{thm:coveringlinear-avertible-algo}). Thus, the time bounds for verification and simple enumeration algorithms does not change in the subset covering model.
In the improved enumeration algorithm we need to additionally address maintaining $R^+_s(W)$, $R^-_t(W)$, and the size of the walk cover of the st-induced \river{} of $W$ during the incremental computation procedure. Since by definition $R^+_s(W)$ (and $R^-_t(W)$) are the same as $R^+(W)$ (and $R^-(W)$) except for the source (and target) of the walks, we can use the same procedure (described in \Cref{sec:algorithms}) to compute them incrementally in total $O(m)$ time. Moreover, the st-induced \river{} of $W$ ($R^+_s(W)\cap R^-_t(W)\cap \river(W)$) depends only on $\river(W)$ and the location of $s$ and $t$ in it, so we can again use the same arguments (as described in \Cref{sec:algorithms}) to prove that each arc enters and leaves this subgraph at most once during the incremental computation. Thus, the size of the $F$ walk-cover can be maintained in $O(m^2)$ time similar to the size of the walk-cover, resulting in the same bound as $E$-covering model.

\begin{theorem}
\label{thm:Fcov}
The bounds of \Cref{thm:improved-linear-runtimes} also hold for the corresponding $F$-covering models.
\end{theorem}

\section{Safety for Subset Visibility}
\label{sec:subsetVisApx}
In subset visibility we limit the solution of the problem to its visible arcs. We thus define $vis_F(W)$ as the substring of a walk $W$ which belongs to $F$. Note that this does not change the solution of the problem but only its representation, which is now limited to the visible set $F$. 
So if a walk $W$ is a solution in the $E$-visible %
model for a problem, then we have $vis_F(W)$ is a solution in the $F$-visible model for the problem, where $F\subseteq E$. 
Abusing notation we also simply say that $W$ is a solution for the  $F$-visible model of the problem.

However, note that the safe solutions ignore the non-visible part of the walk to compute safety. For example \Cref{th:circularSafe} %
relaxes to the visible part of $\vapor(\heart(W))$ to being a single sequence of arcs (possibly disconnected). Henceforth, we shall continue to refer to such a sequence as a visible path despite being disconnected. Thus, relaxing the criteria for safety implies that even though the $F$-visible solution for a problem is the same as the $E$-visible solution, we have that every $E$-visible safe walk is $F$-visible safe but not vice-versa.
As a result, in order to characterise the safety in the $F$-visible model, we relax the definition of the hydrostructure 
by relaxing the forward and backward reachability of a walk as follows:

\begin{definition}
For a walk $W$ where $\wh{}(W),\wt{}(W)\in F$, we define the \emph{restricted} forward and backward  $F$-visible reachability as
    \begin{itemize}[nosep]
        \item $R_F^+(W) = \{x \in G \mid \exists \text{ $\wt{}(W)$-$x$ walk } W' \text{s.t. } vis_F(W) \text{ is not a subwalk of $vis_F(W')$}\}$,
        \item    $R_F^-(W) = \{x \in G \mid \exists \text{ $x$-$\wh{}(W)$ walk } W' \text{s.t. } vis_F(W) \text{ is not a subwalk of $vis_F(W')$}\}$.
    \end{itemize}
\end{definition}

Notice that the modified hydrostructure is merely more relaxed in terms of the \vapor{}. For the $F$-visible model, \Cref{lem:cloud-cases} relaxes the criteria of the \vapor{} being more than a single path, to having a single visible path (see \Cref{fig:f-visible-hears-wings}). Similarly, the properties of the hydrostructure are adapted by considering $vis_F(W)$ instead of $W$. 

\begin{figure}[htb]
    \centering
    \includegraphics[trim=1 3 1 3, clip, scale=1.4]{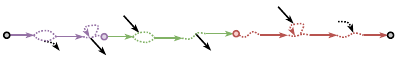}
    \caption{A non-trivial visible walk with visible heart in green, and visible wings in violet and red, where the dotted arcs are invisible.
    }
    \label{fig:f-visible-hears-wings}
\end{figure}

The \textit{visible} wings of a walk $W^l_{F}$ and $W^r_{F}$ and the visible heart $\heart_F(W)$ require \textit{visible splits} and \textit{visible joins} in terms of only visible arcs possibly having invisible arcs (or even subgraphs) between them. Thus, for arcs $e,f\in F$ we consider $(e,f)$ as adjacent if there exists a path from $\head{}(e)$ to $\tail{}(f)$  in $G\setminus F$. %

\paragraph{Implementation.} 

In order to present the verification and enumeration algorithms for any model, it is sufficient to describe the computation of the hydrostructure. The only difference from \Cref{sec:hydrostructure} is that we shall now be using {\em visible} adjacency and hence {\em visible} splits and joins. In order to compute these efficiently we pre-compute the all-pairs reachability among endpoints of arcs in $F$ in $G\setminus F$. This can be computed performing a traversal (BFS or DFS) from the $k$ endpoints of arcs in $F$ in $O(mk)$ time, where $k\leq \min(n,|F|)$. Thus, all the previous algorithms are adaptable in the $F$-visible model with an added expense of $O(mn)$.

\begin{theorem}
\label{thm:Fvis}
The bounds of \Cref{thm:improved-linear-runtimes} also hold for the corresponding $F$-visible models with an added $O(mn)$ term.
\end{theorem}

\end{document}